\documentclass[oneside]{book} 
\usepackage[utf8]{inputenc}
\usepackage{multirow}
\usepackage{hyperref}
\usepackage{amsmath}
\usepackage{amsthm} 
\usepackage{amssymb}

\usepackage[polish,english]{babel}
\usepackage[T1, T2A]{fontenc}

\usepackage{dsfont} 
\usepackage[table,xcdraw]{xcolor}
\usepackage{braket}
\usepackage{fullpage} 
\usepackage{graphicx}

\usepackage{subfig}
\usepackage{float} 
\usepackage{placeins}

\usepackage{multicol}

\newtheorem{theorem}{Theorem}[chapter]
\newtheorem{lemma}{Lemma}[chapter]
\newtheorem{definition}{Definition}[chapter]
\newtheorem{criterion}{Criterion}[chapter]
\newtheorem{fact}{Fact}[chapter]

\newcommand\scalemath[2]{\scalebox{#1}{\mbox{\ensuremath{\displaystyle #2}}}}

\usepackage{natbib}
\usepackage{indentfirst}
\usepackage{braket}

 \usepackage{enumerate}

\title{Quantum correlations in composite systems \\
 under global unitary operations}
\author{Joanna Luc} 
 \date{\today}

\usepackage{indentfirst}

\usepackage{mathtools}

\usepackage[nottoc,numbib]{tocbibind}

\begin{document}
\bibliographystyle{apalike} 

\thispagestyle{empty}
\begin{titlepage}
    \begin{center}

           \Large
	\textbf{Jagiellonian University in Kraków}\vspace{0.2cm}\\
        Faculty of Physics, Astronomy and Applied Computer Science
               \vspace*{1cm}
               
         \vspace{3cm}
         \Large
          \textbf{Joanna Luc}\\\vspace{0.5cm}
         \normalsize Student number: 1053613\\
             \vspace{2cm}
        \Huge
        \textbf{Quantum correlations \\
in composite systems \\ 
under global unitary operations}
      
        \vspace{1.5cm}
        \normalsize
        Master's thesis\\
        in Theoretical Physics\\ \vspace{0.15cm}
        
        \vfill
        \vspace{2cm}
       \begin{minipage}{1\textwidth}
\begin{flushright}
Thesis written under the supervision of\\
prof. dr hab. Karol Życzkowski\\
Institute of Physics, Atomic Optics Department
\end{flushright}
\end{minipage}
        
        \vspace{2cm}
        \begin{center}
      Kraków 2018
        \end{center}
    \end{center}
\end{titlepage}

\tableofcontents

\chapter{Introduction}

Quantum information science investigates quantum effects which can be used for some computational task that is classically impossible or has greater computational complexity. In order to realise such practical goals some more theoretical knowledge is needed, concerning the elements of quantum algorithms, that is, quantum states and operators acting on them. In many cases investigations of this type have theoretical or even foundational importance in themselves. One of such theoretical issues concerns specifically quantum properties of physical systems, especially correlations in bipartite and multipartite states. The topic is very broad: one can investigate relations between these properties, their mathematical measures, practical methods of checking whether a given state has a particular property, and the behaviour of these properties under quantum operations (like unitary evolution or local filtering). This thesis investigates fragments of this wide topic. Its aim is first to give a general picture of properties of quantum states that are analysed in the literature and relations between them. Its more specific goal is to investigate the behaviour of these properties under unitary operations, both global and local. The presented results concern only bipartite states and often their validity is even more restricted, to two-qubit states.

The thesis is organised as follows. In chapter 2 some basic concepts of quantum mechanics and quantum information will be reviewed; its aim is also to fix the notation for the rest of the work.
In chapter 3 properties of quantum states analysed in the literature will be identified. I chapter 4 relations between them will be analysed: whether possessing of one property by a given state implies possessing another property by the same state. In chapter 5 the central notion of the thesis will be introduced, namely possessing a given property in absolute vs. non-absolute way. A property is possessed by a given state in an absolute way if it is preserved under arbitrary unitary operation on this state. I will apply this notion to all properties listed in chapter 3. Two main questions asked here are as follows: Can a given property be possessed both in an absolute way and in a non-absolute way (by different classes of states)? If yes, what are necessary and sufficient conditions for a state possessing a given property in an absolute way?

Most of the thesis is a review of the results already existing in the literature. The original results of this work are contained in Sections \ref{absolute-discord}, \ref{absolute-super-discord}, \ref{absolute-noncontextuality} and Section \ref{sec-relations-absolute}. More specifically, the results contained in Theorems: \ref{thm-absolute-zero-discord}, \ref{thm-absolute-cc-cq-qc}, \ref{thm-absolute-zero-super-discord}, \ref{thm-absolute-product}, \ref{thm:absolute-contextuality} and \ref{relationAB} were obtained by the author and are believed to be new.

The results presented here are of both theoretical and practical interest. From a theoretical point of view, they give us a better understanding of the nature of crucial properties of quantum states, by revealing their relativity (in some cases) to the choice of basis in the Hilbert space. From a practical point of view, unitary operations can be used to transform a given state that lacks a property which is needed for a certain quantum computational task to another state that possesses this property. Therefore, it is important to know the class of states that allow for such a transformation in order to improve our capability of obtaining states useful for a certain quantum computational task from useless ones.

\chapter{Framework}

\section{Basics of Quantum Mechanics}

\subsection{Quantum states}

The basic mathematical object used in quantum mechanics is the $n$-dimensional \textbf{Hilbert space} $\mathcal{H}_n$. In this thesis I will restrict to $n < \infty$ (and usually $n = 4$). Elements of such spaces will be denoted by $\ket{\psi}$ or $\ket{\phi}$, possibly with some indices. For a given Hilbert space $\mathcal{H}$, we denote by $\tilde{\mathcal{H}}$ the space of all operators on $\mathcal{H}$.
Physical objects are described by \textbf{density operators} $\rho \in \tilde{\mathcal{H}}$, $\rho : \mathcal{H} \to \mathcal{H}$ that are hermitian ($\rho^\dagger = \rho$), positive-definite and have unit trace. In the finite cases such operators can be represented by density matrices (denoted by the same symbol). A state is called \textbf{pure} if there exist $\ket{\psi} \in \mathcal{H}$ such that $\rho = \ket{\psi} \bra{\psi}$. Otherwise it is called \textbf{mixed}. All mixed states can be represented as a sum of pure states $\rho = \sum_i p_i \ket{\psi_i} \bra{\psi_i}$, where $0 \leq p_i \leq1$, $\sum_i p_i = 1$. 

A physical system can be \textbf{composed} of two or more \textbf{subsystems}. If this is the case, all the subsystems are described by density operators defined on Hilbert spaces of an appropriate dimension and the whole system is described by a density operator defined on the Hilbert space that is tensor product of all those Hilbert spaces. For example, let us consider a system $\rho^{AB}$ composed of two subsystems $\rho^A$ and $\rho^B$, one of dimension $n$ ($\rho^A : \mathcal{H}_n \to \mathcal{H}_n$) and one of dimension $m$ ($\rho^B : \mathcal{H}_m \to \mathcal{H}_m$). Then a composite system is $\rho^{AB} : \mathcal{H}_n \otimes \mathcal{H}_m \to \mathcal{H}_n \otimes \mathcal{H}_m$. The subsystems can be obtained from a given composite system by taking \textbf{partial trace} of the respective density matrix: $\rho^A = \text{Tr}_B \rho^{AB}$, $\rho^B = \text{Tr}_A \rho^{AB}$. Operation of partial trace is defined by two requirements: for basis states it is $\text{Tr}_B (\ket{\psi_1} \bra{\psi_2} \otimes \ket{\phi_1} \bra{\phi_2}) = \ket{\psi_1} \bra{\psi_2} \text{Tr} (\ket{\phi_1} \bra{\phi_2})$, where $\ket{\psi_i}$ and $\ket{\phi_i}$ are bases of the subsystems; for other cases it is obtained by using the previous formula and linearity.

If the dimension of the Hilbert space is 2, we say of \textbf{one-qubit systems}: $\rho : \mathcal{H}_2 \to \mathcal{H}_2$. The name comes from their use in quantum information, where they are regarded as quantum analogs of classical bits (the smallest units of information). The simplest composite systems are \textbf{two-qubit systems}: $\rho: \mathcal{H}_2 \otimes \mathcal{H}_2 \to \mathcal{H}_2 \otimes \mathcal{H}_2$. The basis for the Hilbert space of one-qubit states (so called computational basis) is denoted by $\{ \ket{0}, \ket{1} \}$. For the Hilbert space of two-qubit states the analogous basis is $\{ \ket{00} = \ket{0} \otimes \ket{0}$, $\ket{01} = \ket{0} \otimes \ket{1}$, $\ket{10} = \ket{1} \otimes \ket{0}$, $\ket{11} = \ket{1} \otimes \ket{1} \}$. Most of the results available in the quantum information literature concern two-qubit systems, because of their simplicity on the one hand (for higher dimension calculations usually become much more complicated), and non-triviality on the other (computational gains require some quantum correlations, which are present only in composite systems).

\subsection{Properties of quantum states}

Quantum states differ with each other in many ways and there are many formal tools to describe these differences. In this work the following properties will be considered: nonlocality/locality, steerability/unsteerability, entanglement/separability, PPT (Positive Partial Transpose) property, negative/non-negative conditional entropy, non-zero/zero quantum discord (connected with the properties of being classical-classical, quantum-classical and classical-quantum state), non-zero/zero quantum super discord (connected with the property of being a product state) and contextuality/noncontextuality. In each of these pairs one property is the opposite of the other, so a given state possesses exactly one property from each pair.

\subsection{Operations on quantum states}

Quantum systems can evolve in time and this evolution is described by various \textbf{quantum operations}. Let us denote the initial system by $\rho$ and the final system after the operation by $\mathcal{E} (\rho)$, so the operation itself is denoted by $\mathcal{E}$. There are three almost equivalent approaches to define quantum operations \cite[ch. 8.2]{nielsen-chuang}:

\begin{enumerate}
\item \textbf{System coupled to environment}: 

We assume that the entire physical system can be divided into principal system $\rho$, in which we are interested, and the environment $\rho_{\text{env}}$. We also assume that initially the entire system is in a product state: $\rho \otimes \rho_{\text{env}}$. The system as a whole is isolated, so its evolution is described by a unitary operator $U$. However, the principal system can interact with its environment, so its evolution is not necessarily unitary. We are interested in the final state of the principal system only, so the environment should be traced out and we obtain $\mathcal{E} (\rho)= \text{Tr}_{\text{env}} \left[ U (\rho \otimes \rho_{\text{env}}) U^\dagger \right]$.

\item \textbf{Operator sum representation}: 

Every quantum operation $\mathcal{E}$ can be represented as a sum $\mathcal{E} = \sum_{k} E_k \rho E_k^\dagger$, where Kraus operators $E_k$ satisfy the relation $\sum_k E_k^\dagger E_k \leq \mathds{1}$, which follows from the requirement that the trace cannot increase, Tr $\mathcal{E} (\rho) \leq 1$. 

\item \textbf{Physically motivated axioms}: 

Quantum operation $\mathcal{E}$ is a map from the space of density operators of the system $Q_1$ to the space of density operators of the system $Q_2$ that satisfies the following three conditions:
\begin{itemize}
\item $\mathcal{E}$ is a trace-non-increasing map: for every density operator $\rho$, $0 \leq \text{Tr} \mathcal{E} (\rho) \leq 1$,
\item $\mathcal{E}$ is a convex-linear map: for probabilities $\{ p_i \}$ and density operators $\rho_i$, $\mathcal{E} \left( \sum_i p_i \rho_i \right) = \sum_i p_i \mathcal{E} (\rho_i)$,
\item $\mathcal{E}$ is a completely positive map: for any density operator $\rho$, $\mathcal{E}(\rho)$ is positive and also for any auxiliary system $R$, $(\mathds{1}_R \otimes \mathcal{E}) (\rho')$ is positive, where $\mathds{1}_R$ is an identity operator on the system $R$ and $\rho'$ is a state of a joint system $Q_1 R$. 
\end{itemize}
\end{enumerate}

The three approaches are not fully equivalent, because the second and the third one allow for non-trace-preserving operations, for which $0\leq$ Tr $\mathcal{E} (\rho)< 1$, whereas the first approach allows only for trace-preserving operations, for which Tr $\mathcal{E} (\rho)= 1$. From the physical point of view, non-trace-preserving operations are needed for the description of measurement.



\subsection{Unitary operations}

One of the most important classes of quantum operators are unitary operators, satisfying the condition of unitarity $U U^\dagger = U^\dagger U = \mathds{1}$, where $\mathds{1}$ is an identity operator. They describe time evolution of isolated quantum systems. On the other hand, they can be viewed as a change of a basis. Call these two views 'dynamic' and 'static' interpretation of unitary operations, respectively. They are not competitive: this distinction means only that two physically different operations (change in time, which is physically real and change of a basis, which is only a formal manipulation) are represented by the same mathematical operation. The choice of an interpretation depends on a situation that is analysed.

For composite systems one can distinguish between local and global unitary operations. \textbf{Local unitary operations} have a form $U = U_A \otimes U_B$, where $U_A$ and $U_B$ are unitary operators that act independently in each subsystem. Such operations cannot change correlations between subsystems. \textbf{Global unitary operations} do not have this form and therefore intertwine both subsystems, possibly changing correlations between them.

\subsection{Measurement}

Another important type of quantum operations is quantum measurement. There are different types of measurements that may be performed on quantum states, including the von Neumann measurement, the generalised measurement and the weak measurement, defined below. For the details of the first two see e.g. \citealp{nielsen-chuang}.

\begin{definition}  \label{von-Neumann-measurement}
\textbf{Projective measurement} (von Neumann measurement) is described by an observable, $M$, which is a hermitian operator on the state space of the observed system. The observable has a spectral decomposition
\begin{equation}
M = \sum_{m} m P_m,
\end{equation}
where $P_m$ is the projector onto the eigenspace of $M$ with eigenvalue $m$. Projectors satisfy $\sum_m P_m = \mathds{1}$ and $P_m P_{m'} = \delta_{m m'}$. 
\end{definition}

The system that before the projective measurement was in a state $\ket{\psi}$, after the measurement changes its state to
\begin{equation}
\ket{\psi'} = \frac{P_m \ket{\psi}}{\sqrt{p(m)}}
\end{equation}

\begin{definition}  [\citealp{nielsen-chuang}]
\textbf{Generalised measurement} (POVM -- Positive Operator Valued Measure) is described by a collection of measurement operators $M_m$ that satisfy $\sum_m M^{\dagger}_m M_m = \mathds{1}$ (but  are not necessarily projectors). 
\end{definition}

A special case of the generalised measurement is the weak measurement, first introduced in the paper \citealp{aharonov-1988}. With a view to its application in context of quantum super discord, instead of the original definition I will use the following one \cite {oreshkov-2005}:
\begin{definition} \label{weak-measurement}
Weak measurement is given by a pair of operators: 
\end{definition}
\begin{equation}
P (\xi) = \sqrt{\frac{1 - \text{tanh} \xi}{2}} \Pi_1 + \sqrt{\frac{1 + \text{tanh} \xi}{2}} \Pi_2 ,
\end{equation}
\begin{equation}
P (-\xi) = \sqrt{\frac{1 + \text{tanh} \xi}{2}} \Pi_1 + \sqrt{\frac{1 - \text{tanh} \xi}{2}} \Pi_2 ,
\end{equation}
\textit{where $\Pi_1$ and $\Pi_2$ are two orthogonal projectors satisfying $\Pi_1 + \Pi_1 = \mathds{1}$ and $|\xi| \ll 1$ is the strength of the measurement.}

\medskip

The presented definition of weak measurement is not the most general one but is in some sense universal: it can be proven that any projective measurement and any generalised measurement can be decomposed into a sequence of measurements of this type \cite {oreshkov-2005}. In general one can consider the broader range of $x$, namely $x \in [0, \infty )$. Such operators have the following properties \cite{super-discord-1}:
\begin{itemize}
\item $P^\dagger (\xi) P (\xi) + P^\dagger (-\xi) P(-\xi) = \mathds{1}$,
\item $P (0) = \frac{1}{\sqrt{2}} \mathds{1}$, so that for $\xi=0$ weak measurement does not change a quantum state at all,
\item $\text{lim}_{\xi \to \infty} P(-\xi) = \Pi_1$, $\text{lim}_{\xi \to \infty} P(\xi) = \Pi_2$, hence in the limit $\xi \to \infty$ this operation becomes the projective measurement.
\end{itemize}

The physical sense of the von Neumann measurement is relatively clear: it 'detects' the value of a given physical quantity (observable) of a quantum system in a given state and projects this state into eigenstate of this observable associated with the measured value. This 'detecting' is usually not deterministic in the sense that more than one value can be obtained in a given measurement with non-zero probability; but these values and probabilities are uniquely determined by the measured quantity and the state of the system. 
Generalised measurements (excluding von Neumann measurements, which are a special case) do not reveal these values and therefore do not give us precise information about a physical system. They are a mathematical representation of detectors with non-ideal efficiency, measurement outcomes that include additional randomness, measurements that give incomplete information, etc. The same is true for weak measurements. The paper which introduces weak measurements is in this context symptomatic: it proves that using this special type of measurements one can obtain arbitrarily large outcome when measuring a component of the spin of a spin$-\frac{1}{2}$ particle. The motivation for weak measurement will be explained in Section \ref{sec:super-quantum-discord}, devoted to super quantum discord, definition of which uses this notion.

One can also define a measurement performed on a single part of a composite system:
\begin{definition} \label{local-von-Neumann-measurement}
For a bipartite state $\rho^{AB}$, a \textbf{local von Neumann measurement} on a subsystem $A$ is a family of
one-dimensional orthogonal projections on the space of subsystem $A$, $\{ \Pi^A_i \}$, such that $\sum_i \Pi^A_i = \mathds{1}^A$. 
\end{definition}

\begin{definition} \label{local-von-Neumann-measurement}
For a bipartite state $\rho^{AB}$, a \textbf{local weak measurement} on a subsystem $A$ is a pair of weak measurement operators acting on a subsystem $A$, $\{ P^A (x), P^A (-x) \}$.
\end{definition}

Needless to say, analogous definitions can be formulated for a subsystem $B$.

\section{Entropy}

\subsection{Classical entropy of a probability vector} \label{sec:classical-entropy}

Classical probability of some random variable $A$ can be described by a vector of probabilities $p(a)$ of obtaining the particular values $a$ of this variable. For such vectors we can define entropy and some derivative notions, which measure its information content:
\begin{itemize}
\item Shannon entropy: $H(A) = - \sum_{a} p(a) \log p(a)$,
\item Joint entropy: $H (A, B) = - \sum_{a, b} p(a, b) \log p(a, b)$,
\item Conditional entropy: $H (A | B) = - \sum_{a, b} p(a, b) \log p(a | b)$,
\item Mutual information: $I(A:B) = H(A) + H(B) - H(A, B)$.
\end{itemize}

Everywhere in this thesis by '$\log$' we understand logarithm to base 2. One can show the following relations between conditional entropy on the one hand, and the joint entropy and entropies of the random variables considered separately on the other: 
\begin{equation} \label{eq:classical-conditional-entropy-relation1} H(A |B) = H (A, B) - H(B), \end{equation}
\begin{equation} \label{eq:classical-conditional-entropy-relation2} H (B | A) = H(A, B) - H(A). \end{equation}

Another important fact about classical conditional entropies are its bounds, especially its lower bound:
\begin{equation} 0 \leq H(A|B) \leq H(A). \end{equation}

\subsection{Entropy of quantum states} \label{sec:quantum-entropy}

Quantum probability is encoded in density operators describing quantum systems. It is possible to define quantum analogons of classical entropies presented in the section \ref{sec:classical-entropy} (see e.g. \citealp{nielsen-chuang}, ch. 11). The analogy is rather straightforward, with the exception of conditional entropy, for which the definition is based on the relations \eqref{eq:classical-conditional-entropy-relation1} and \eqref{eq:classical-conditional-entropy-relation2}. The analogon of Shannon entropy concerning a quantum state $\rho$ is called von Neumann entropy $S(\rho)$. The derivative notions have the same names as in the classical case. Therefore, we obtain the following list of quantum entropies:

\begin{itemize}
\item von Neumann entropy: $S (\rho) = - \text{Tr} (\rho \log \rho)$,
\item Joint entropy: $S(\rho^{AB}) = - \text{Tr} (\rho^{AB} \log \rho^{AB})$,
\item Conditional entropy: $S(\rho^A | \rho^B) = S(\rho^{AB}) - S(\rho^B)$,
\item Mutual information: $I(\rho^A : \rho^B) = S(\rho^A) + S(\rho^B) - S(\rho^{AB})$.
\end{itemize}

Let us elaborate more on what information is provided by these quantities.
\begin{itemize}
\item Von Neumann entropy of a state describes the minimal amount of
information necessary to fully specify this state.
\item The joint entropy is the entropy of the entire system.
\item The conditional entropy between two systems is the entropy of the system
  minus any information gained from the other system from the
  correlations they share. 
\item Quantum mutual information describes amount of
information contained in the joint state that exceeds the information
locally available to $A$ and $B$. It is interpreted as a measure of total correlations in a state
$\rho^{AB}$, both quantum and classical. Two systems are correlated
if together they contain more information than taken separately. 
\end{itemize}

Quantum entropies can be better understood when we look at the picture illustrating relations between them, so called Venn diagram for the entropies:

\begin{figure}[H]
\centering
\includegraphics[scale=0.3]{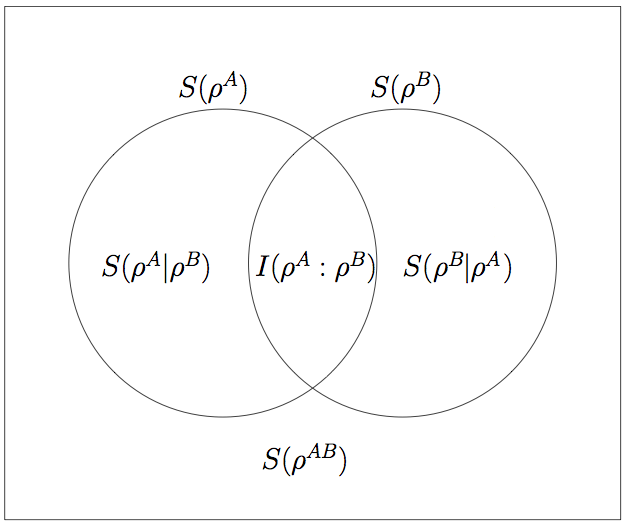}
\caption{Venn diagram for quantum entropies. Source: \citealp{vedral2006}.}
\end{figure}

There is a well-known theorem that describes the relation between von Neumann entropy and its classical counterpart, Shannon entropy:

\begin{theorem} [\citealp{von-neumann-1927}] \label{vonneumann-shannon}
For any quantum state $\rho$ with spectrum $\vec{d} = d_1, \ldots, d_n$, its von Neumann entropy is equal to the Shannon entropy of the spectrum, $S (\rho) = H(\vec{d})$.
\end{theorem}

%

In contrast to its classical counterpart, quantum conditional entropy can be less than zero. More specifically, its bounds are: $- S(B) \leq S(A|B) \leq S(A)$. However one needs to remember that in the quantum case the basic entropy, that is, von Neumann entropy, is always positive and conditional entropy is only some algebraic combination of such entropies, not an entropy in a proper sense. 

There exist in the literature other types of entropy, which are generalisations of von Neumann entropy (see e.g. \citealp[ch. 13.4]{bengtsson-zyczkowski}) One of them is the family of Rényi $\alpha$-entropies $S_{\alpha} (\rho) = \frac{1}{1 - \alpha} \log \text{Tr} (\rho^{\alpha})$ parametrized by $\alpha \geq 1$, so that for $\alpha \rightarrow 1$ this quantity converges to the von Neumann entropy $S(\rho)$. Generalised entropies allow to define analogous derivative notions, including joint entropy, conditional entropy and mutual information. In \cite{friis} properties of conditional Rényi entropy and conditional von Neumann entropy are compared. In this thesis generalised entropies are omitted, therefore 'entropy' always means the 'von Neumann entropy' here.

\section{Fano-Bloch decomposition}

Any bipartite state $\rho^{AB}$ of dimension $d_A \times d_B$ can always be represented in the so called Fano-Bloch form \cite{fano-1983}:

\begin{equation} \label{eq:fano-bloch}
\rho_{AB}  = \frac{1}{d_A d_B}  \left( \mathds{1}_{AB} + 
\sum_{m=1}^{d_A^2 -1} a_m \sigma^A_m \otimes \mathds{1}_B 
+ \sum_{n=1}^{d_B^2 - 1} \mathds{1}_A \otimes b_n \sigma^B_n + \sum_{m=1}^{d_A^2 - 1}
    \sum_{n=1}^{d_B^2 - 1} t_{mn} \sigma^A_m \otimes \sigma^B_n
\right),
\end{equation}
where
$a_m = \text{Tr} \rho^{AB} (\sigma^A_m \otimes \mathds{1}_B)$, $b_n = \text{Tr} \rho^{AB} (\mathds{1}_A \otimes \sigma^B_n)$ are Bloch vectors of reduced states $\rho^A, \rho^B$,
$t_{mn} = \text{Tr} \rho_{AB} (\sigma^A_m \otimes \sigma^B_n)$ is a  correlation tensor and
$\sigma^A_m$, $\sigma^B_n$ are generalised Pauli matrices satisfying $\text{Tr} (\sigma^i_m \sigma^i_n) = 2 \delta_{mn}$, $\text{Tr} (\sigma^i_n) = 0$, where $i = A, B$. These conditions mean that matrices associated with a given subsystem are orthogonal and that all of them are traceless.

\section{Special classes of quantum states} \label{sec:special-states}

Some classes of states have a special status because they have particularly simple form (allowing for substantial simplifications in calculations), while still being non-trivial. Two examples relevant for our purposes are the Weyl states (that encompass the the Bell states and the Werner states) and the Gisin states.

\subsection{Weyl states}

The first class of states which will be used here are the locally maximally mixed states, also known as the Weyl states. They satisfy
$a_m = b_n = 0$ and 
$\rho_A = \frac{\mathds{1}_A}{d_A}, \rho_B = \frac{\mathds{1}_B}{d_B}$.
For $d_A = d_B = 2$ the Weyl states can be represented, up to local unitaries, as
\begin{equation} \label{eq:weyl}
\rho_{Weyl} = \frac{1}{4} \left( \mathds{1}_{AB}
+ \sum_{n = 1}^{3} \tilde{t}_n \sigma^A_n \otimes \sigma^B_n
\right) = \left(
\begin{array}{cccc}
 \frac{\tilde{t}_3+1}{4} & 0 & 0 & \frac{\tilde{t}_1-\tilde{t}_2}{4} \\
 0 & \frac{1-\tilde{t}_3}{4} & \frac{\tilde{t}_1+\tilde{t}_2}{4} & 0 \\
 0 & \frac{\tilde{t}_1+\tilde{t}_2}{4} & \frac{1-\tilde{t}_3}{4} & 0 \\
 \frac{\tilde{t}_1-\tilde{t}_2}{4} & 0 & 0 & \frac{\tilde{t}_3+1}{4} \\
\end{array}
\right).
\end{equation}

One well-known type of the Weyl states are the Bell states, which taken together form the maximally entangled state form a basis in the Hilbert space of two-qubit states,:

\begin{equation}
\ket{\Phi^+} = \frac{1}{\sqrt{2}} \left( \ket{00} + \ket{11} \right),
\end{equation}
\begin{equation}
\ket{\Phi^-} = \frac{1}{\sqrt{2}} \left( \ket{00} - \ket{11} \right),
\end{equation}
\begin{equation}
\ket{\Psi^+} = \frac{1}{\sqrt{2}} \left( \ket{01} + \ket{10} \right),
\end{equation}
\begin{equation}
\ket{\Psi^-} = \frac{1}{\sqrt{2}} \left( \ket{01} - \ket{10} \right).
\end{equation}

The Bell states can be expressed in the following form, from which one can easily see that they are a special case of the Weyl states:

\begin{equation} \label{eq:bell1}
\rho_{\Phi^+} = \ket{\Phi^+} \bra{\Phi^+} = \frac{1}{4} \left(
  \mathds{1}_4 + \sigma_x \otimes \sigma_x - \sigma_y \otimes \sigma_y
  + \sigma_z \otimes \sigma_z
\right),
\end{equation}
\begin{equation} \label{eq:bell2}
\rho_{\Phi^-} = \ket{\Phi^-} \bra{\Phi^-} = \frac{1}{4} \left(
  \mathds{1}_4 - \sigma_x \otimes \sigma_x + \sigma_y \otimes \sigma_y
  + \sigma_z \otimes \sigma_z
\right),
\end{equation}
\begin{equation} \label{eq:bell3}
\rho_{\Psi^+} = \ket{\Psi^+} \bra{\Psi^+} = \frac{1}{4} \left(
  \mathds{1}_4 + \sigma_x \otimes \sigma_x + \sigma_y \otimes \sigma_y
  - \sigma_z \otimes \sigma_z
\right),
\end{equation}
\begin{equation}\label{eq:bell4}
\rho_{\Psi^-} = \ket{\Psi^-} \bra{\Psi^-} = \frac{1}{4} \left(
  \mathds{1}_4 - \sigma_x \otimes \sigma_x - \sigma_y \otimes \sigma_y
  - \sigma_z \otimes \sigma_z
\right).
\end{equation}

From equations \eqref{eq:weyl} and \eqref{eq:bell1}-\eqref{eq:bell4} follows a well-known fact describing the relation between the Weyl states and the states diagonal in the Bell basis.

\begin{fact} \label{weyl-bell-diagonal} 
A state is diagonal in the Bell basis, that is, it is of the form:
\begin{equation}
\rho = p_1 \rho_{\Phi^+} + p_2 \rho_{\Phi^-} + p_3 \rho_{\Psi^+} + p_4 \rho_{\Psi^-} 
\end{equation}
with $p_i \geq 0$, $\sum_i p_i = 1$ iff it is a Weyl state of the form \eqref{eq:weyl}.
\end{fact}

Note that not all Weyl states are of the form \eqref{eq:weyl}, so without this specification the theorem does not hold.
Another special case of the Weyl states are the Werner states, introduced in \citealp{werner1989}, where they were used to prove that entanglement does not imply nonlocality (some entangled Werner states are local). The Werner states are convex combinations of one of the maximally entangled states and the maximally mixed state:

\begin{equation} \label{eq:werner}
\rho_{\text{Werner}}  = w \rho_{\Psi^-} + \frac{1}{4} ( 1 - w )\mathds{1}_4
= \frac{1}{4} \left( \mathds{1}_2 \otimes \mathds{1}_2 - w
  \vec{\sigma} \otimes \vec{\sigma}
\right).
\end{equation}

The particularly important fact about the Werner states is that they are parametrized by a single parameter $w$, which takes values $w \in [ -\frac{1}{3}, 1 ]$. On the other hand, this family of states is theoretically nontrivial, as it encompasses some entangled states, the maximally mixed state and the spectrum of intermediate states. 

\begin{figure}[H]
\centering
\includegraphics[scale=0.3]{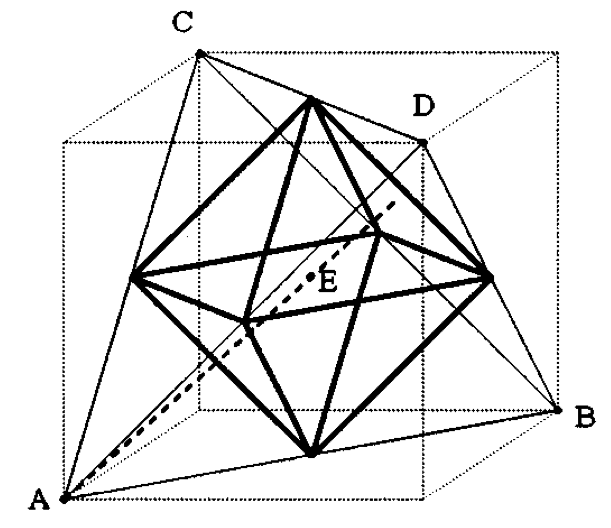}
\caption{Geometrical representation of Weyl states. All the physical states are contained in the tetrahedron. The axes are parameterised by variables $\tilde{t}_1, \tilde{t}_2, \tilde{t}_3$ present in equation \eqref{eq:weyl}. The
    bold-line-contoured octahedron represents separable states, the
    dashed line denotes the set of the Werner states. Here $A = (-1, -1, -1)$, $B
  = (1, 1, -1)$, $C = (1, -1, 1)$, $D = (-1, 1, 1)$ (correspond to
  Bell states: $\ket{\Psi^-}, \ket{\Psi^+}, \ket{\Phi^+}, \ket{\Phi^-}$), and $E$ denoted the normalised
  identity. Source: \citealp{horodecki-1996}.}
\end{figure}

\subsection{Gisin states} \label{sec:gisin}

Another interesting class of states, the Gisin states, was introduced in \citealp{gisin1996}. They have been used to prove existence of the phenomenon of hidden nonlocality: some of the Gisin states which are local (do not violate CHSH inequality) lose this feature after applying a purely local operations, so called local filtering. The Gisin states are expressed by the following formula:
\begin{equation} \label{eq:gisin}
\rho_G (\lambda, \theta) = \lambda \rho_{\theta} + (1 - \lambda ) \rho_{\text{top}},
\end{equation}
where
\begin{equation} 
\rho_{\theta} = \ket{\psi_{\theta}} \bra{\psi_{\theta}},
\ket{\psi_{\theta}} = \sin \theta \ket{01} + \cos \theta \ket{10},
\end{equation}
\begin{equation} 
\rho_{\text{top}} = \frac{1}{2} \left(
\ket{00} \bra{00} + \ket{11} \bra{11}
\right)
\end{equation}
and the ranges of the parameters are $0 \leq \theta \leq \frac{\pi}{2}$\footnote{Originally the Gisin states were defined only for $0 < \theta < \frac{\pi}{2}$, but the broader range of the parameter $\theta$ does not spoil their physicality and will be useful in our analysis of some properties and some absolute properties of the Gisin states in Sections \ref{sec-relations-special} and \ref{sec-relations-absolute-special}.}, $ 0 \leq \lambda \leq 1 $.

In matrix form the Gisin states look as follows:
\begin{equation}
\rho_G = \left(
\begin{array}{cccc}
 \frac{1-\lambda }{2} & 0 & 0 & 0 \\
 0 & \lambda  \sin ^2 \theta  & \lambda  \sin \theta  \cos \theta  & 0 \\
 0 & \lambda  \sin \theta  \cos \theta  & \lambda  \cos ^2 \theta  & 0 \\
 0 & 0 & 0 & \frac{1-\lambda }{2} \\
\end{array}
\right).
\end{equation}

Their eigenvalues are: $0, \frac{1-\lambda}{2}, \frac{1-\lambda}{2}, \lambda$. Observe, that they are dependent only on the parameter $\lambda$ and are independent of the second parameter $\theta$. This fact will be important later, in Section \ref{sec-relations-absolute-special}, where we will use these eigenvalues to check which the Gisin states possess some absolute properties. The Gisin states are not the subset of the Weyl states, but these two classes of states have a non-empty intersection, as the Gisin states are locally maximally mixed for $\theta = \frac{\pi}{4}$.

\chapter{Properties of quantum states}

\section{Separability vs. entanglement}

Let $\rho^{AB} \in \tilde{\mathcal{H}}_A \otimes \tilde{\mathcal{H}}_B$ be a bipartite state of an arbitrary dimension. It always belongs to at least one of the following three types of states:

\begin{definition}
A bipartite state $\rho^{AB}$ is said to be a \textbf{product} state iff there exist
$\rho^A$ and $\rho^B$ such that $\rho^{AB} = \rho^A \otimes
\rho^B$. 
\end{definition}

\begin{definition} \label{def-separability}
A bipartite state $\rho^{AB}$ is called \textbf{separable} iff there exist states $\rho^A_k$ and
$\rho^B_k$ and numbers $p_1, \ldots, p_r$, $p_k > 0$, $\sum_{k=1}^r
p_k = 1$ such that $\rho^{AB} =
\sum_{k=1}^r p_k \rho^A_k \otimes \rho^B_k$ (separable state is a
convex combination of product states).
\end{definition}

\begin{definition} \label{def-entanglement}
If a state is not separable, then it is called \textbf{entangled}.
\end{definition}

Note that the above three notions gain their meaning only after the composed system is splitted into two subsystems, called $A$ and $B$, so that the tensor product $\tilde{\mathcal{H}}_A \otimes \tilde{\mathcal{H}}_B$ is defined. Between the above properties the following relation holds: any product state is separable but (in general) not vice versa. For pure states, separability and being a product state coincide. 
Entanglement is the most important and most widely analysed feature of quantum states. It is recognised as a resource for most of quantum computation tasks, in the sense that it is responsible for specifically quantum effects in these tasks. For a review of current research concerning quantum entanglement see \citealp{horodeccy-entanglement}. 

It is possible to check whether a given state is a product state. Let us consider two qubit states:
\begin{equation} \label{eq:sub-product}
\rho_1 = \left(
\begin{array}{cc}
a & b \\
c & d \\
\end{array}
\right), \quad
\rho_2 = \left(
\begin{array}{cc}
e & f \\
g & h \\
\end{array}
\right).
\end{equation}
Their tensor product, which is a two-qubit state, looks as follows:
\begin{equation} \label{eq:product}
\rho = \rho_1 \otimes \rho_2 = \left(
\begin{array}{cccc}
ae & af & be & bf \\
 ag & ah  & bg  & bh \\
ce & cf  & de  & df \\
cg & ch & dg & dh \\
\end{array}
\right).
\end{equation}
Now, we can reverse this reasoning and formulate the following necessary condition for being a product state:
\begin{criterion} \label{criterion-product}
If a two-qubit state 
\begin{equation}
\rho =  \left(
\begin{array}{cccc}
 m_{11} & m_{12} & m_{13} & m_{14} \\
m_{21} & m_{22} & m_{23} & m_{24} \\
m_{31} & m_{32} & m_{33} & m_{34} \\
 m_{41} & m_{42} & m_{43} & m_{44} \\
\end{array}
\right)
\end{equation}
is a product state, then its matrix elements satisfy the following conditions:
\begin{multicols}{3}
\noindent
\centering
\begin{enumerate}
\item $m_{11}m_{14} = m_{12} m_{13}$, 
\item $m_{11} m_{23} = m_{21} m_{13}$, 
\item $m_{11} m_{24} = m_{22} m_{13}$, 
\item $m_{12} m_{23} = m_{21} m_{14}$, 
\item $m_{12} m_{24} = m_{22} m_{14}$,
\item $m_{21} m_{24} = m_{22} m_{23}$,
\item $m_{11} m_{32} = m_{12} m_{31}$,
\item $m_{11} m_{41} = m_{21} m_{31}$,
\item $m_{11} m_{42} = m_{22} m_{31}$,
\item $m_{12} m_{41} = m_{21} m_{32}$,
\item $m_{12} m_{42} = m_{22} m_{32}$,
\item $m_{21} m_{42} = m_{22} m_{42}$,
\item $m_{11} m_{34} = m_{12} m_{33}$,
\item $m_{11} m_{43} = m_{21} m_{33}$,
\item $m_{11} m_{44} = m_{22} m_{33}$,
\item $m_{12} m_{43} = m_{21} m_{34}$,
\item $m_{12} m_{44} = m_{22} m_{34}$,
\item $m_{21} m_{44} = m_{22} m_{43}$,
\item $m_{13} m_{32} = m_{14} m_{31}$,
\item $m_{13} m_{41} = m_{23} m_{31}$,
\item $m_{13} m_{42} = m_{24} m_{31}$,
\item $m_{14} m_{41}= m_{23} m_{32}$,
\item $m_{14} m_{42} = m_{24} m_{32}$,
\item $m_{23} m_{42} = m_{24} m_{41}$,
\item $m_{13} m_{34} = m_{14} m_{33}$,
\item $m_{13} m_{43} = m_{23} m_{33}$,
\item $m_{13} m_{44} = m_{24} m_{33}$,
\item $m_{14} m_{43} = m_{23} m_{34}$,
\item $m_{14} m_{44} = m_{24} m_{34}$,
\item $m_{23} m_{44} = m_{24} m_{43}$,
\item $m_{31} m_{34} = m_{32} m_{33}$,
\item $m_{31} m_{43} = m_{41} m_{33}$,
\item $m_{31} m_{44} = m_{42} m_{33}$,
\item $m_{32} m_{43} = m_{41} m_{34}$,
\item $m_{32} m_{44} = m_{42} m_{34}$,
\item $m_{41} m_{44} = m_{42} m_{43}$.
\end{enumerate}
\end{multicols}
\end{criterion}

\begin{proof}
Assume that $\rho$ is a product state. Then there exist $\rho_1$ and $\rho_2$ as in equation \eqref{eq:sub-product} and the relation between them and the state $\rho$ is given by \eqref{eq:product}. Now, the above 36 conditions are equivalent to the following equalities in terms of $a, b, c, d, e, f, g, h$:
\begin{multicols}{3}
\noindent
\centering
\begin{enumerate}
\item $ae \cdot bf = af \cdot be$, 
\item $ae \cdot bg = ag \cdot be$, 
\item $ae \cdot bh = ah \cdot be$, 
\item $af \cdot bg = ag \cdot bf$, 
\item $af \cdot bh = ah \cdot bf$,
\item $ag \cdot bh = ah \cdot bg$,
\item $ae \cdot cf = af \cdot ce$,
\item $ae \cdot cg = ag \cdot ce$,
\item $ae \cdot ch = ah \cdot ce$,
\item $af \cdot cg = ag \cdot cf$,
\item $af \cdot ch = ah \cdot cf$,
\item $ag \cdot ch = ah \cdot cg$,
\item $ae \cdot df = af \cdot de$,
\item $ae \cdot dg = ag \cdot de$,
\item $ae \cdot dh = ah \cdot de$,
\item $af \cdot dg = ag \cdot df$,
\item $af \cdot dh = ah \cdot df$,
\item $ag \cdot dh = ah \cdot dg$, 
\item $be \cdot cf = bf \cdot ce$,
\item $be \cdot cg = bg \cdot ce$,
\item $be \cdot ch = bh \cdot ce$,
\item $bf \cdot cg = bg \cdot cf$,
\item $bf \cdot ch = bh \cdot cf$,
\item $bg \cdot ch = bh \cdot cg$,
\item $be \cdot df = bf \cdot de$,
\item $be \cdot dg = bg \cdot de$,
\item $be \cdot dh = bh \cdot de$,
\item $bf \cdot dg = bg \cdot df$,
\item $bf \cdot dh = bh \cdot df$,
\item $bg \cdot dh = bh \cdot dg$,
\item $ce \cdot df = cf \cdot de$,
\item $ce \cdot dg = cg \cdot de$,
\item $ce \cdot dh = ch \cdot de$,
\item $cf \cdot dg = cg \cdot df$,
\item $cf \cdot dh = ch \cdot df$,
\item $cg \cdot dh = ch \cdot dg$.
\end{enumerate}
\end{multicols}
\end{proof}

\section{PPT (Positive Partial Transpose) states}

Definitions of separability and entanglement involve universal quantification over the set of density matrices and weights, so on the basis of this definition alone it is difficult to check whether a given state is separable or entangled. Therefore, it would be helpful to find some simpler criterions of separability and entanglement. Such criterions were found only for some special cases and a general criterion, working for any state of arbitrary dimension, is still not known. For \textbf{pure} bipartite states the criterion is given by the following theorem:

\begin{criterion} 
A state $\rho^{AB}$ is separable iff the entropy of the reduced state is positive, $S(\rho^A) > 0$ (equivalently: $S(\rho^B) > 0$).
\end{criterion}

%

For \textbf{mixed} bipartite states there is no universal criterion that gives necessary and sufficient conditions for separability. However, for Hilbert spaces of dimension $2 \times 2$ or $2 \times 3$ such a conditions are given by the \textbf{PPT} (Positive Partial Transpose) criterion. The following definition of partial transpose allows one to formulate the PPT criterion:

\begin{definition}
For a given state $\rho$, its partial transpose with respect to a subsystem $B$, $\rho^{\top_B} $, is given by $\bra{m} \bra{\mu} \rho^{\top_B} \ket{n} \ket{\nu} := \bra{m} \bra{\nu} \rho \ket{n} \ket{\mu}$. Analogously for a subsystem $A$.
\end{definition}

\begin{criterion} [\citealp{ppt}, \citealp{horodecki-1996b}] \label{criterion-ppt}
A quantum state $\rho$ of dimension $2 \times 2$ or $2 \times 3$  is separable iff $\rho^{\top_B} \geq 0$ (equivalently: $\rho^{\top_A} \geq 0$).
\end{criterion}

%

Positive Partial Transpose criterion gives a sufficient condition for Hilbert spaces of all finite dimensions, but only for dimensions $2 \times 2$ and $2 \times 3$ it gives necessary condition as well. For higher dimensions there exists a class of states that have positive partial transpose despite being entangled; they belong to the class of bound entangled states (for more details see \citealp{horodeccy-entanglement} and references therein).

\section{Locality vs. nonlocality}

Some of the famous works in foundations of quantum mechanics concern quantum nonlocality. Nonlocality was discovered by Bell (see his collected papers in \citealp{bell}) and later analysed by many others, including \citealp{chsh}. The phenomenon can be described as follows (see also Fig. \ref{fig:bell}): the source ($S$) produces physical system, which is divided into two subsystems. They are send to two distant observers, called Alice and Bob. Upon receiving their subsystems, each observer performs a measurement on it. The measurement chosen by Alice is labeled $x$ and its outcome is $a$. Similarly, Bob chooses measurement $y$ and gets outcome $b$. The experiment is characterised by the joint probability distribution $p (ab | xy)$ of obtaining outcomes $a$ and $b$ when Alice and Bob choose measurements $x$ and $y$, respectively. It turns out that the joint probability distribution predicted by quantum mechanics in general is not a product of probability distributions obtained by Alice and Bob considered separately: $p(a,b | x,y) \neq p(a | x) p(b | y)$, so these distributions are not independent, irrespectively of how large is the distance between the observers. One may wonder whether this independence is real or the quantum-mechanical description is incomplete and it is possible to introduce an additional factor, so called hidden variable, which enables one to describe the two subsystems as uncorrelated. The second option has been explored under the name of hidden variable models for quantum systems. In fact, possessing such a model is the defining condition for the state to be local.

\begin{definition} \label{def:locality}
A bipartite state $\rho^{AB}$ is called \textbf{local} iff it can be described by local hidden variable model, that is, there exist a hidden variable $\lambda \in \Lambda$ and a probability measure $\mu$ on the space $\Lambda$ such that for every measurement choices $x, y$ one can reconstruct joint probability distribution $p (ab | xy)$ predicted by quantum mechanics from another probability distribution conditionalised on $\lambda$:
\begin{equation}
p (a, b | x, y) = \int_{\Lambda} d \lambda \ q (\lambda) p(a,b | x,y, \lambda),
\end{equation}
where
\begin{equation} 
p(a,b | x,y, \lambda) = p(a | x, \lambda) p(b | y, \lambda) .
\end{equation}
\end{definition}

\begin{definition} \label{def:nonlocality}
A bipartite state $\rho^{AB}$ is called \textbf{nonlocal} iff it is not local.
\end{definition}

\begin{figure}[H]
\centering
\includegraphics[scale=0.4]{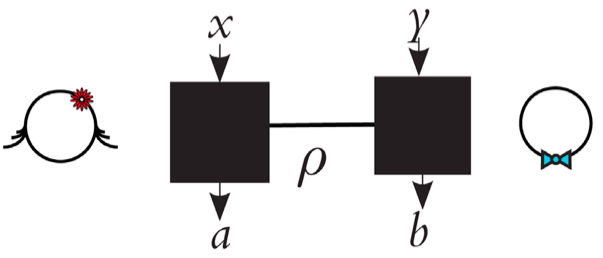}
\caption{Sketch of the Bell experiment. Source: \citealp{calvacanti-2017}.}
\label{fig:bell}
\end{figure}

Bell observed that existence of hidden variable model leads to some constraints on probabilities of the outcomes, which are known under the name of the Bell inequality. There are several versions of this inequality and the most popular is CHSH (Clauser-Horne-Shimony-Holt) inequality \cite{chsh}, which is more general than Bell's original one. Let us assume that $a, b \in \{ -1, +1 \}$ and define expectation value of joint measurement of values $a$ and $b$ with measurement choices $x$ and $y$:
\begin{equation}
\langle a_x b_y \rangle = \sum_{a, b} a b \ p (a, b | x, y).
\end{equation}
It can be proven that states which are nonlocal are precisely these states which violate the \textbf{CHSH inequality}: 
\begin{theorem} [\citealp{chsh}]
A bipartite state $\rho^{AB}$ is nonlocal iff for some settings $\vec{a}, \vec{a}', \vec{b}, \vec{b}'$ it violates the CHSH inequality
\begin{equation}
\langle a b \rangle + \langle a b' \rangle + \langle a' b \rangle - \langle a' b' \rangle \leq 2.
\end{equation}
\end{theorem}

In general Bell-type scenarios are characterised by three numbers: the number of subsystems, the number of possible measurements, and the number of outcomes of each measurement. Here we are interested in the CHSH inequality, which concerns the scenario with two subsystems, two measurements and two outcomes; therefore, it is (2,2,2)-type Bell inequality. With different types of scenarios there are connected different types of nonlocality, but from now on we will use this term to denote only one particular kind of nonlocality, namely CHSH-nonlocality. 


The above results can be also formulated in terms of the formalism of density operators. We need to use the \textbf{CHSH operator}, defined as follows:
\begin{equation} \label{eq:CHSH-operator}
\mathfrak{B}_{\text{CHSH}} := \vec{a} \cdot \vec{\sigma}^A \otimes (\vec{b} +
\vec{b}') \cdot \vec{\sigma}^B + \vec{a}' \cdot \vec{\sigma}^A \otimes
(\vec{b} - \vec{b}') \cdot \vec{\sigma}^B
\end{equation}
and Hilbert-Schmidt inner product: $(A|B)_{HS} := \text{Tr} (A^{\dagger} B)$ .

\begin{theorem} 
A bipartite state $\rho^{AB}$ is nonlocal iff for some settings $\vec{a}, \vec{a}', \vec{b}, \vec{b}'$ it violates the inequality
\begin{equation}
(2 \mathds{1} - \mathfrak{B}_{\text{CHSH}} | \rho )_{HS} \geq 0,
\end{equation}
where $\mathfrak{B}_{\text{CHSH}}$ is given by the equation \eqref{eq:CHSH-operator}.
\end{theorem}

As in the case of separability and entanglement, the definition of nonlocality involves a quantification over a large set. Thus, relying only on the definition it is difficult to check whether a given state is local or not. Again, no universal criterion providing relatively simple necessary and sufficient conditions for nonlocality is known. However, for two-qubit states the following criterion was found:
\begin{criterion}[CHSH operator criterion, \citealp{horodeccy-1995}] \label{criterion-chsh}
Let $\rho$ be the density operator of a two-qubit state with correlation tensor $t = ( t_{mn} )$, defined in \eqref{eq:fano-bloch}, and let $\mu_1$ and $\mu_2$ be the
two largest eigenvalues of $M_{\rho} = t^\top t$. The state is \textbf{nonlocal} iff 
\begin{equation}
\max_{\vec{a}, \vec{a}', \vec{b}, \vec{b}'} \langle
\mathfrak{B}_{CHSH} \rangle = 2 \sqrt{\mu_1 + \mu_2} > 2.
\end{equation}
\end{criterion}

\section{Quantum steering}

The concept of quantum steering was introduced in \cite{schrodinger-1935}. It is the name for the fact that one of the parties ($A$ or $B$) can change the state of the other ($B$ or $A$, respectively) by choosing a basis for local measurement (the state of the second party must collapse according to this choice). If a bipartite state allows steering, it is called steerable. In contrast to nonlocality and entanglement, this property of quantum states is not symmetric between $A$ and $B$ ($A$'s being steerable by $B$ is something different from $B$'s being steerable by $A$). The notion was defined mathematically in \citealp{wiseman-2007}. For a review of recent research on this topic see \citealp{calvacanti-2017}. Similarly to nonlocality scenarios, steering scenarios can be characterised by the number of subsystems (here we restrict to two), the number of possible measurements and the number of their outcomes. The definition is also similar --- instead of local hidden variables models, it uses the notion of local hidden state models, which can be described roughly as local hidden variable models for one subsystem only.

\begin{definition}
A bipartite state $\rho^{AB}$ is said to be \textbf{steerable} from A to B (B to A) if there exists measurements in Alice’s (Bob’s) part that produces an assemblage that does not admit local hidden state model, that is, there exist no hidden variable $\lambda \in \Lambda$ and no probability measure $\mu$ on the space $\Lambda$ such that
\begin{equation}
\sigma_{a | x} = \int_{\Lambda} d \lambda \ \mu (\lambda) p (a, x | \lambda) \rho_{\lambda}^B,
\end{equation}
where $\sigma_{a | x} = p(a | x) \rho_{a | x}$, $\rho_{a | x} = \mathrm{Tr}_A (M_{a | x} \otimes \mathds{1}) \rho^{AB} / p(a | x)$ are Bob's states after Alice's measurement (with the setting $x$ and the outcome $a$) and $p (a | x)$ are probabilities of these states.
\end{definition}

The following theorem describes the relationship between steerability and nonlocality of the type (2,2,2):
\begin{theorem} [\citealp{girdhar-calvacanti-2016}]
If a two-qubit state $\rho$ is steerable with CHSH-type measurements, i.e., with a set-up (2,2,2), then it violates CHSH inequality.
\end{theorem}

From this theorem it follows that for the $(2,2,2)$ case nonlocality and steerability are equivalent. However, there are some states which are steerable with three measurements but not CHSH-nonlocal (for the examples see \citealp{calvacanti-2017}). Necessary and sufficient conditions for steerability are in general not known. The exception is a two-qubit case, for which necessary and sufficient conditions are analysed in \cite{nguyen-vu} and \cite{yu-2018}. For other dimensions there are some partial results, e.g. many inequalities providing sufficient conditions for steerability are derived in \cite{calvacanti-2009}. The following inequality provides condition for steerability of type (2,3,n) of a bipartite state $\rho^{AB}$ (as analysed e.g. in \citealp{hindusi-steering}):
\begin{equation} \label{eq:3-steerability}
\frac{1}{\sqrt{3}} \left| \sum_{i=1}^3 \langle A_i \otimes B_i \rangle \right| \leq 1,
\end{equation}
where $A_i = \vec{a}_i \vec{\sigma}$, $B_i = \vec{b}_i \vec{\sigma}$, $\vec{a}_i$ and $\vec{b}_i$ are measurement directions, $\vec{\sigma} = (\sigma_1, \sigma_2, \sigma_3)$ is a vector composed of Pauli matrices, and $\langle A_i \otimes B_i \rangle  = \text{Tr} (\rho^{AB} A_i \otimes B_i)$

\begin{figure}[H]
\centering
\includegraphics[scale=0.4]{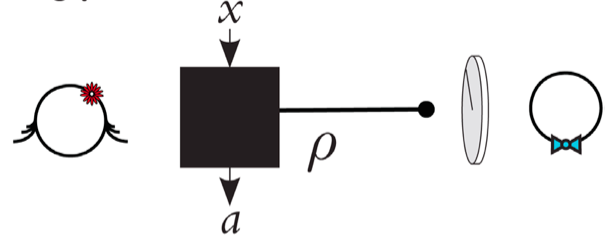}
\caption{Quantum steering experiment. There are two parties, Alice (left) and Bob (right). In contrast to nonlocality scenario, only Alice performs an experiment, choosing measurement $x$ and obtaining outcome $a$. Source: \citealp{calvacanti-2017}.}
\label{fig:bell}
\end{figure}

\section{Negative vs. non-negative conditional entropy}

As mentioned in Section \ref{sec:classical-entropy}, conditional entropy of a quantum state can be negative. The physical meaning of such a phenomenon is described in \citealp{horodecki-2005}. Conditional entropy provides the answer to the following question: Given an unknown quantum state distributed over two systems, how much quantum communication is needed to transfer the full state to one system? If the conditional entropy is positive, its sender needs to communicate this number of quantum bits to the receiver; if it is negative, then sender and receiver instead gain the corresponding potential for future quantum communication. These intuitions are formalised in the so called quantum state merging protocol, whose details can be found in the mentioned paper.

\section{Quantum discord}

\subsection{Definition and meaning}

Quantum discord was introduced in \cite{ollivier-zurek} as a new measure of quantum correlations that encompasses broader class of states than entanglement. The definition is as follows:

\begin{definition}
\textbf{Quantum entropy} of a state $\rho^{AB}$ \textbf{with respect to measurement on the subsystem A},
$\{\Pi^A_i \}$, is $S(\rho^B | \{\Pi^A_i \}) = \sum_i p_i S (\rho^{B |  \Pi^A_i})$, 
where $p_i = \mathrm{Tr} ((\Pi^A_i \otimes \mathds{1}_B) \rho^{AB})$, and $\rho^{B | \Pi^A_i} = \mathrm{Tr}_A ((\Pi^A_i \otimes \mathds{1}_B) \rho^{AB})/ p_i$.
\end{definition}

\begin{definition}
Quantum \textbf{ discord} of a state $\rho^{AB}$ under a measurement on the subsystem $A$, $\{ \Pi^A_i \}$, is the difference $D (B |A) := I (B:A) - J(B : A)$,
where $ I (B:A) $ is a mutual information defined in Section \ref{sec:quantum-entropy}, $J(B : A) = max_{\{ \Pi^A_i \}} J(B | \{ \Pi^A_i \}) $, $J(B | {\{ \Pi^A_i \}}) := S(B) - S(B | \{ \Pi^A_i \})$.
\end{definition}


Classical counterparts of $I$ and $J$ coincide: $I_{cl} (A : B):= S(A) + S(B)
- S (A, B) = S(B) - S(B | A) =: J_{cl} (B : A)$ and this is why the quantity has been called 'discord'. 
To shed light on the physical meaning of quantum discord, it would be useful to quote researches investigating this quantity:

\begin{quote}
Classical information is locally accessible, and can be obtained
without perturbing the state of the system: One can interrogate just
one part of a composite system and discover its state while leaving
the overall density matrix (as perceived by observers that do not have
access to the measurement outcome) unaltered. A general separable $\rho$ does not allow for such insensitivity to measurements: Information can
be extracted from the apparatus but only at a price of perturbing $\rho$, even when this density matrix is separable. However, when
discord disappears, such insensitivity (which may be the defining
feature of “classical reality,” as it allows acquisition of
information without perturbation of the underlying state) becomes
possible for correlated quantum systems. \cite{ollivier-zurek}
\end{quote}

\begin{quote}
Quantum discord is the minimum part of the mutual information shared
between A and B that cannot be obtained by the measurement on A. \cite[p. 68]{lectures-correlations}
\end{quote}

\subsection{Properties of quantum discord}

Quantum discord can be shown to possess the following properties \cite{bera2018}:

\begin{itemize}
\item quantum discord is nonnegative $D(B|A) \geq 0$,
\item in general quantum discord is not symmetric $D(B|A) \neq D(A|B)$,
\item quantum discord is invariant under local unitary transformations,
\item for bipartite pure states quantum discord reduces to entropy of entanglement,
\item any entangled state has a non-zero discord,
\item quantum discord is bounded by the entropy $D(B|A) \leq S(B)$.
\end{itemize}

It was argued \cite{datta-2008} that a non-zero quantum discord of a given state indicates its usefulness for quantum computation, sometimes even in the absence of entanglement. An example is so called DQC1 (deterministic quantum computation with single qubit) protocol. The task in this protocol is to compute a trace of a unitary matrix. The authors argue that computational advantage over classical protocols does not depend on entanglement.

\subsection{When quantum discord is zero?} \label{zero-discord}

There is no general formula for computing quantum discord even for two-qubit states. Only results for special classes of states are available.
There exist analytic results for Weyl states \cite{luo-discord-weyl} and also for broader class of states, so called X-states, that is, the states which have non-zero values only on their diagonal and anti-diagonal positions in the computational basis \cite{ali-2010}. 
Checking whether a given quantum state has zero discord is much easier than computing quantum discord in the case when it is non-zero. There exist at least two relatively simple criteria for checking whether a given bipartite state has zero discord. They are as follows:

\begin{criterion} [\citealp{dakic}] \label{criterion-zero-discord-1}
If $\rho^{AB}$ is a two-qubit state, then $\rho^{AB}$ has both quantum discords equal to zero ($D(A|B) = D(B|A) =0$) iff $|| \vec{x} ||^2 + ||t||^2
  -k_{max} = 0$, where $k_{max}$ is the largest eigenvalue of matrix $K =
  \vec{x} \vec{x}^\top + tt^\top$, $||t||^2 = \mathrm{Tr} \ t^\top t$, and $\vec{x}$, $t = (t_{\mu \nu})$ are defined by equation \eqref{eq:fano-bloch}. 
\end{criterion}

For Weyl states this criterion gives $\tilde{t}_2^2 + \tilde{t}_2^2 + \tilde{t}_3^2
  - \mathrm{max}\{ \tilde{t}_2^2, \tilde{t}_2^2, \tilde{t}_3^2 \} = 0$, where $\tilde{t}_i$, $i = 1, 2, 3$ are defined in equation \eqref{eq:weyl}.

\begin{criterion} [\citealp{huang2011}] \label{criterion-zero-discord-2}
A bipartite quantum state $\rho^{AB} \in \tilde{\mathcal{H}}_A \otimes \tilde{\mathcal{H}}_A$ has zero quantum discord, $D(A|B) = 0$, iff all the square blocks of its density matrix of dimension $d = \mathrm{dim} (\mathcal{H}_B)$ are normal matrices and commute with each other. For $D(B|A) =0$ one needs to consider all the blocks of dimension $d = \mathrm{dim} (\mathcal{H}_A) $.
\end{criterion}

\begin{figure}[h]
\centering
\includegraphics[scale=0.33]{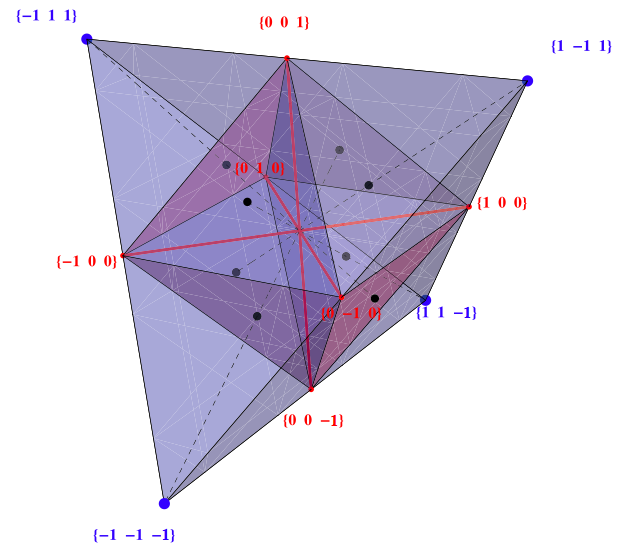}
\caption{Geometrical representation of Weyl states. The axes are parameterised by variables $\tilde{t}_1, \tilde{t}_2, \tilde{t}_3$ present in equation \eqref{eq:weyl}. All the physical states are contained in the tetrahedron. The
    bold-line-contoured octahedron represents separable states. The zero-discord
  states are labeled by the solid red lines. Source: \citealp{dakic}.}
\label{zero-discord-weyl}
\end{figure}
\FloatBarrier

The first of these criterions is formulated only for two-qubit states, whereas the second one works for arbitrary bipartite states. The first criterion does not distinguish between $D(A|B)$ and $D(B|A)$; it follows that although in general $D(A|B) \neq D(B|A)$, for two-qubit states if one of them is equal to zero, then the other one must also be equal to zero. The illustration of application of this criterion to the family of Weyl states is provided in Fig. \ref{zero-discord-weyl}. From the second criterion one may conclude that it is possible that exactly one of discords vanishes only if the dimension of system $A$ is different than the dimension of system $B$. In fact, there exist quantum states which have only one discord equal to zero.

I will mention here one more criterion, which gives only sufficient condition for zero quantum discord, but is interesting because it allows to estimate the relative volume of the set of states with zero discord.

\begin{criterion} [\citealp{ferraro-2010}]
If $[ \rho^{AB}, \rho^A \otimes
  \mathds{1}_B ] = 0$, then $D(B|A) = 0$ (but not the other way
  around --- e.g. all pure maximally entangled states)
\end{criterion}

It can be proven \cite{ferraro-2010} that the set of density matrices satisfying $[ \rho^{AB}, \rho^A \otimes \mathds{1}_B ] = 0$ has measure zero in the set of all density matrices. Therefore, the set of zero discord states is also of measure zero. For comparison: separable pure states are of measure zero in the  set of all pure states, but separable states have a positive measure in the set of all density matrices \cite{zyczkowski-1998}. Therefore, there are significantly less state with zero quantum discord than states which are separable. This fact can be understood as an indication that quantum discord is able to 'detect' more quantum correlations than the properties described earlier.

\subsection{Classical-classical, quantum-classical and classical-quantum states}

One can define three interesting classes of states, which turn out to be strictly connected with the notion of quantum discord. These are classical-classical, quantum-classical and classical-quantum states, defined as follows (see e.g. \citealp{lectures-correlations}):

\begin{definition} \label{def:classical-classical}
A state $\rho^{AB}$ is called \textbf{classical-classical} iff it has a
form $\rho^{AB} = \sum_{i,j} p_{ij}^{AB} \Pi^A_i \otimes \Pi^B_j$, where $\{ p_{ij}^{AB} \}$ is a
classical probability distribution, $\Pi^A_i := \ket{i}_A \bra{i}$ and $\Pi^B_j := \ket{j}_B
\bra{j}$  are spectral projections of the reduced states
$\rho^A = tr_B \rho^{AB}$ and $\rho^B = tr_A \rho^{AB}$, respectively, $\{ \ket{i}
\}$ and $\{ \ket{j} \}$ are orthonormal bases for parties $A$ and $B$, respectively.
\end{definition}

\begin{definition} \label{def:classical-quantum}
A state $\rho^{AB}$ is called \textbf{classical-quantum} iff it has a form
$\rho^{AB} = \sum_{i} p_{i}^{A}  \Pi^A_i \otimes \rho^B_i $. 
\end{definition}

\begin{definition} \label{def:quantum-classical}
A state $\rho^{AB}$ is called \textbf{quantum-classical} iff it has a form
$\rho^{AB} = \sum_{j} p_{j}^{B} \rho^A_j \otimes \Pi^B_j$. 
\end{definition}

It is useful to compare these definitions with the Definition \ref{def-separability} of separable states. All of them postulate similar forms of states: they should be sums of tensor products of states of subsystems. The difference lies in the details of the form of states of subsystems: sometimes we require that they should be spectral projections of the reduced states (in the case of classical-classical for both subsystems, in the case of quantum-classical and classical-quantum for one subsystem --- called 'classical'), and sometimes we do not impose on them any additional conditions (in the case of separable states for for both subsystems, in the case of quantum-classical and classical-quantum for one subsystem --- called 'quantum').

One can prove the theorem connecting the above classes of states with quantum discord (see e.g. \citealp{bera2018}):

\begin{theorem} \label{thm:classicality-vs-discord} The following equivalences hold:

\noindent
A bipartite state $\rho^{AB}$ is classical-classical iff $D(A|B) = D(B|A) =0$. 

\noindent
A bipartite state $\rho^{AB}$ is classical-quantum iff $D(B|A) =0$.

\noindent
A bipartite state $\rho^{AB}$ is quantum-classical iff $D(A|B) =0$.
\end{theorem}

\section{Quantum super discord} \label{sec:super-quantum-discord}

The notion of super quantum discord was introduced in \citealp{super-discord-1}. It is similar to the notion of quantum discord --- the only difference lies in the fact that it uses weak measurements (see Definition \ref{weak-measurement}) instead of the standard von Neumann measurements (see Definition \ref{von-Neumann-measurement}). First let us look at how the authors state their motivation for considering weak measurements.

After the local von Neumann measurement on a subsystem $A$ a bipartite state $\rho^{AB}$ collapses to a classical-quantum state and after the local von Neumann measurement on a subsystem $B$ it collapses to a quantum-classical state. Therefore, this kind of measurement destroys the correlations in the system $\rho^{AB}$, at least these accessible from the point of view of the subsystem on which we have measured. (Recall that we are considering measurements on a single subsystem, not on an entire system. In the second case it is possible that a post-measurement state will be entangled, namely when the basis in which we measure contains entangled states.) In contrast, after a weak measurement a system can still be in an entangled state. That is where the name of this type of measurement comes from: the weak measurement destroys a state in a lesser degree than the standard von Neumann measurement. This phenomenon can be illustrated by the example given by \citealp{super-discord-1}. Let us consider maximally entangled pure state $\rho_{\Phi^+} $ given by \eqref{eq:bell1}. After a weak measurement on the subsystem $B$ this state is transformed into
\begin{equation}
\rho = \frac{1}{2} \left[ \ket{00} \bra{00} + \ket{11} \bra{11} + \text{sech} x (\ket{00} \bra{11} + \ket{11} \bra{00})
\right].
\end{equation}
It can be shown \cite{super-discord-2} that the above output state is still entangled for sufficiently small values of $x$. Super quantum discord measures the correlation in a state $\rho^{AB}$ as seen by an observer who performs a weak measurement on one of the subsystems. Now, let us state the formal definition of super quantum discord (with respect to subsystem $A$), taken from  \citealp{super-discord-1}:

\begin{definition}
\textbf{Quantum entropy} of a state $\rho^{AB}$ \textbf{with respect to weak measurement on the subsystem A},
$\{P^A (\pm  \xi) \}$, is $S(\rho^B | \{P^A (\xi) \}) =  p(\xi) S (\rho^{B |  P^A (\xi)} ) + p(-\xi) S (\rho^{B |  P^A (-\xi)} ) $, 
where $p (\pm \xi) = \mathrm{Tr} ((P^A (\pm \xi) \otimes \mathds{1}_B) \rho^{AB})$, $\rho^{B | P^A (\pm \xi)} = \mathrm{Tr}_A ((P^A (\pm \xi) \otimes \mathds{1}_B) \rho^{AB})/ p (\pm \xi)$
\end{definition}

\begin{definition}
Quantum \textbf{super discord} of a state $\rho^{AB}$ under a weak measurement on subsystem $A$, $\{ P^A (\pm x) \}$,  is a difference $D (B |A) := I (B:A) - J(B : A)$,
where $J(B : A) = max_{\xi} J(B | \{ P^A (\pm \xi) \}) $, $J(B | {\{ P^A (\pm \xi) \}}) := S(B) - S(B | \{ P^A (\pm \xi) \})$.
\end{definition}

In the above definitions $\xi$ is fixed, so the sums contain only two components, for $\xi$ and for $-\xi$. To understand better the physical meaning of this quantity let us quote the authors who invented the above definition:
\begin{quote} 
A remarkable feature of the super quantum discord is that for pure entangled states it can exceed the quantum entanglement. In this sense, SQD reveals quantum correlation that truly goes beyond quantum entanglement even for pure entangled states. (...) Thus, quantum correlation is not only observer dependent but also depend on how gently or strongly one perturbs the quantum system. \cite{super-discord-1}
\end{quote}

The following properties of super quantum discord will be interesting for us:

\begin{theorem} [\citealp{super-discord-1}] \label{thm:super-discord-vs-discord} 
For any bipartite state $\rho^{AB}$, the quantum super discord is greater than or equal to the quantum discord: $D_w (A|B) \geq D(A|B)$.
\end{theorem}

\begin{theorem} [\citealp{super-discord-2}]\label{thm:super-discord-product} 
A bipartite $\rho^{AB}$ has zero super quantum discord $D_w (A|B) = D_w (B|A) =0$ iff $\rho^{AB}$ is a product state.
\end{theorem}

\begin{theorem} [\citealp{super-discord-2}]
A bipartite $\rho^{AB}$ has zero super quantum discord $D_w (A|B) = D_w (B|A) =0$ iff $\rho^{AB}$ has zero mutual information $I(A:B) \equiv I(B:A) =0$.
\end{theorem}

In the case of quantum discord it is possible to have $D(A|B) = 0$ and $D(B|A) \neq 0$ or the other way around (see Section \ref{zero-discord}). In contrast, if a super quantum discord is zero with respect to one subsystem, it is also zero with respect to the other subsystem.

\section{Contextuality vs. noncontextuality}

The last pair of properties to be analysed in this thesis is contextuality and noncontextuality. In general, noncontextuality means that the measured value of any observable is independent on other observables that are measured jointly with it. Of course we restrict only to observables that are compatible with a given observable (i.e. commuting with it), because otherwise they could not be measured jointly. There are two senses of contextuality: it can be understood as state-independent property of a set of projectors \cite{kochen-specker} and as state-dependent property, which is possessed by some states but not the others. In this thesis I will be interested only in the second sense of contextuality. It can be formalised in terms of nonexistence of contextual hidden variable theory, which leads to certain inequality (in analogy to nonlocality). There are many versions of this inequality with different numbers of projectors and the best known of them is KCBS (Klyachko-Can-Binicioğlu-Shumovsky) inequality \cite{kcbs}. First, following \citealp{kitajima-2017}, let us formally define the notions of contextuality and noncontextuality:

\begin{definition}
A state $\ket{\psi}$ is noncontextual iff there exist a hidden variable $\lambda \in \Lambda$, a probability measure $\mu$ on the space $\Lambda$ and a value assignment on observables that can be measured on $\rho$, i.e. a function $\nu : \mathds{A} \times \Lambda \mapsto \mathds{R}$ satisfying for any two commuting observables $A, B$:
\begin{enumerate}
\item $\nu (A + B | \lambda) = \nu (A | \lambda) + \nu (B | \lambda)$,
\item $\nu (A B | \lambda) = \nu (A | \lambda) \nu (B \lambda)$,
\item $\nu (\mathds{1} | \lambda) = 1$,
\item $\nu (0 | \lambda) = 0$,
\item $\bra{\psi} A \ket{\psi} = \int_{\Lambda} \nu (A | \lambda) \mu (\lambda) d \lambda$.
\end{enumerate}
\end{definition}

\begin{definition}
A state $\ket{\psi}$ is contextual iff it does not satisfy noncontextuality condition.
\end{definition}

\begin{theorem} [\citealp{kcbs}]
A state $\ket{\psi}$ is noncontextual iff for any family of projectors $P_0$, $P_1$, $P_2$, $P_3$, $P_4$ such that each $P_i$ commutes with $P_{i+1}$ (where the sum should be understood modulo 5), the KCBS inequality holds:
\begin{equation} \label{eq:kcbs}
\bra{\psi} (P_0 + P_1 + P_2 + P_3 + P_4) \ket{\psi} \leq 2.
\end{equation}
\end{theorem}

According to the theorem, a state is contextual iff for some family of projectors satisfying conditions specified above the KCBS inequality \eqref{eq:kcbs} is violated, that is, $\bra{\psi} (P_0 + P_1 + P_2 + P_3 + P_4) \ket{\psi} > 2 $. 
There are known examples both of states that are contextual and of states that are noncontextual. The simplest noncontextual state is identity operator \cite{kitajima-2017}. Examples of contextual states are provided in \cite{jerge}; that paper contains also results of experimental tests that confirm violation of KCBS inequality.

\chapter{Relations between different properties}

\section{Relations for special classes of states} \label{sec-relations-special}

In previous chapter many properties of quantum states were described. One may ask then a question, whether they are really different from each other (they are not equivalent, in other words: the sets of states that possess them are not equal) and if yes, what are relations between them (whether some of them imply some other properties, in other words: whether sets of states possessing these properties are included in each other). This chapter gives the answer to this question. First look at the properties of two classes of states introduced in Section \ref{sec:special-states}.

In Table \ref{table-non-absolute} there are shown ranges of parameters for which the Werner states \eqref{eq:werner} and the Gisin states \eqref{eq:gisin} are product states, have zero discord, are separable, unsteerable, local and have non-negative conditional entropy. They are also illustrated in Fig. \ref{fig:werner} and \ref{fig:gisin}. The results have been obtained with the use of Criteria \ref{criterion-product}, \ref{criterion-ppt}, \ref{criterion-chsh} and \ref{criterion-zero-discord-2}. Some of these results were already presented in the literature: \citealp{werner1989} (separability and locality for the Werner states), \citealp{gisin1996} (separability and locality for the Gisin states), \citealp{friis} (non-negative conditional entropy for the Werner states and for the Gisin states), \citealp{ollivier-zurek} (quantum discord for the Werner states). 

For the Werner states the ranges of parameters are different with the exception of product and zero discord states. However, the Gisin states are product and zero discord for different ranges of parameters, therefore these two families of states are sufficient to distinguish between all the mentioned properties completely. From these results it follows that no two of the mentioned properties are equivalent.

\begin{figure}[H]
\centering
\includegraphics[scale=0.7]{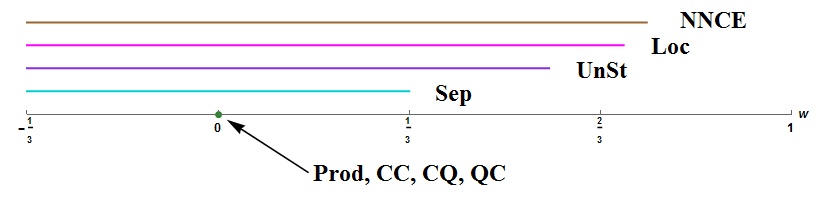}
\caption{Properties of the Werner states. The only product state (Prod) is the state with $w=0$ (green point); it is also the only state with zero discord (and therefore classical-classical CC, classical-quantum CQ and quantum-classical QC). The other properties are: separability (Sep, blue), unsteerability (UnSt, violet), locality (Loc, purple) and non-negative conditional entropy (NNCE, brown).}
\label{fig:werner}
\end{figure}

\begin{figure}[H]
\centering
\includegraphics[scale=0.7]{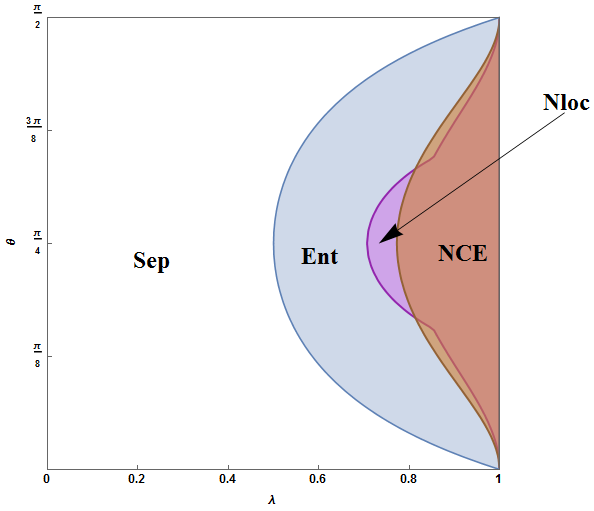}
\caption{Properties of the Gisin states \eqref{eq:gisin} in the space of parameters $\lambda, \theta$. The shaded regions are: NLoc --- nonlocal (violet), NCE --- negative conditional entropy (brown), Ent --- entangled (blue). The white region contains all and only separable Gisin states (Sep). Observe that there are states nonlocal and with negative conditional entropy, states local and with negative conditional entropy, as well as states nonlocal with non-negative conditional entropy, so there is clearly no implication between nonlocality and negative conditional entropy. This illustration is based on \cite{friis}.}
\label{fig:gisin}
\end{figure}

\begin{table}[H]
\centering
{\def\arraystretch{2}\tabcolsep=5pt
\begin{tabular}{llll}
\hline
\multicolumn{1}{|l|}{}                             & \multicolumn{1}{l|}{\textbf{Werner states} \eqref{eq:werner}}                               & \multicolumn{2}{l|}{\textbf{Gisin states} \eqref{eq:gisin}}                                                                                                                                                                                                                                                                                                                                              \\ \hline
\multicolumn{1}{|l|}{parameter}                    & \multicolumn{1}{l|}{$w$}                                         & \multicolumn{1}{l|}{$\theta$}                                                                                                                                                                     & \multicolumn{1}{l|}{$\lambda$}                                                                                                                                             \\ \hline
\multicolumn{1}{|l|}{range of the parameter}       & \multicolumn{1}{l|}{$w \in [-\frac{1}{3},1]$}                          & \multicolumn{1}{l|}{$\theta \in [0, \frac{\pi}{2}]$}                                                                                                                                                         & \multicolumn{1}{l|}{$\lambda \in [0, 1]$}                                                                                                                                              \\ \hline
\multicolumn{1}{|l|}{product}                      & \multicolumn{1}{l|}{$w=0$}                                           & \multicolumn{1}{l|}{$\theta \in \{ 0, \frac{\pi}{2} \}$}                                                                                                                                          & \multicolumn{1}{l|}{$\lambda = 1$}                                                                                                                                         \\ \hline
\multicolumn{1}{|l|}{zero discord}                 & \multicolumn{1}{l|}{$w=0$}                                           & \multicolumn{1}{l|}{$\theta$ arbitrary}                                                                                                                                                                    & \multicolumn{1}{l|}{$\lambda = 0$}                                                                                                                                         \\ \cline{3-4} 
\multicolumn{1}{|l|}{}                             & \multicolumn{1}{l|}{}                                            & \multicolumn{1}{l|}{$\theta \in \{ 0, \frac{\pi}{2} \}$}                                                                                                                                          & \multicolumn{1}{l|}{$\lambda$ arbitrary}                                                                                                                                             \\ \hline
\multicolumn{1}{|l|}{separable}                    & \multicolumn{1}{l|}{$w \in [-\frac{1}{3}, \frac{1}{3}]$}                 & \multicolumn{2}{l|}{\parbox[t]{7cm}{$\lambda \cos^2 \theta \geq 0$ \& $\lambda \sin^2 \theta \geq 0$ \& \\
$1 - \lambda (1 + \sin (2 \theta)) \geq 0$ \& $1 - \lambda (1 - \sin (2 \theta)) \geq 0$ \\
$ $}}           \\ \hline
\multicolumn{1}{|l|}{unsteerable}                  & \multicolumn{1}{l|}{$w \in [-\frac{1}{3}, \frac{1}{\sqrt{3}}]$}                             & \multicolumn{1}{l|}{?}                                                                                                                                                                             & \multicolumn{1}{l|}{?}                                                                                                                                                      \\ \hline
\multicolumn{1}{|l|}{local}                        & \multicolumn{1}{l|}{$w \in [-\frac{1}{3}, \frac{1}{\sqrt{2}}]$}         & \multicolumn{2}{l|}{max$\left\{ \sqrt{\lambda^2 \sin^2(2\theta)+(1-2\lambda)^2}, \sqrt{2} \lambda \sin(2 \theta)\right\} \leq 1$}                                                                                                                                                                                                                                                             \\ \hline
\multicolumn{1}{|l|}{\parbox[t]{3cm}{non-negative \\ conditional entropy}} & \multicolumn{1}{l|}{$w \in [-\frac{1}{3}, w_0]$, $w_0 \approx 0.7476$} & \multicolumn{2}{l|}{\parbox[t]{7cm}{$-2 \frac{1-\lambda}{2} \log \left( \frac{1-\lambda}{2} \right) - \lambda \log \lambda $\\
$+ \left( \frac{1-\lambda}{2} + \lambda \cos^2 \theta \right) \log \left( \frac{1-\lambda}{2} \lambda \cos^2 \theta  \right) $\\
$+ \left( \frac{1-\lambda}{2} + \lambda \sin^2 \theta  \right) \log \left( \frac{1-\lambda}{2} + \lambda \sin^2 \theta  \right)  > 0 $\\
$ $}} \\ \hline                                                                                                                                                           
\end{tabular}%
}
\caption{Selected properties of the Werner states and the Gisin states in function of their parameters.}
\label{table-non-absolute}
\end{table}

\section{General relations}

We have seen that no two properties analysed here are equivalent. However, at least for the Werner states and the Gisin states there are some implications between them. One may ask whether these implications are specific to this classes of states or they hold in general. It turns out that in some cases the answer is positive and some cases is negative. For example, we have already seen that although for the Werner states negative conditional entropy implies nonlocality, this is no longer true for the Gisin states. The following theorem summarises what is known about these relations in the general case:

\begin{theorem} \label{thm-implications-1}
For any bipartite state $\rho$ the following implications hold: 
\begin{enumerate}
\item $\rho$ is nonlocal $\Rightarrow$ $\rho$ is steerable,
\item $\rho$ is steerable with $n$ settings $\Rightarrow$ $\rho$ is steerable with $n+1$ settings,
\item there exist $n \geq 2$ such that $\rho$ is steerable with with $n$ settings $\Rightarrow$ $\rho$ is entangled,
\item $\rho$ is entangled $\Rightarrow$ $\rho$ is not classical-quantum $D(B|A) \neq 0$,
\item $\rho$ is entangled $\Rightarrow$ $\rho$ is not quantum-classical $D(A|B) \neq 0$,
\item $\rho$ is not quantum-classical $\Rightarrow$ $\rho$ is not classical-classical,
\item $\rho$ is not classical-quantum $\Rightarrow$ $\rho$ is not classical-classical,
\item $\rho$ is not classical-classical $\Rightarrow$ $\rho$ is not a product state, i.e. $\rho$ has non-zero super quantum discord $D_w (A|B) \neq0$, $D_w (B|A) \neq 0$.
\end{enumerate}
\end{theorem}

\begin{proof} $ $ \\
\begin{enumerate}
\item See \cite{calvacanti-2017}.
\item If one has $n+1$ settings at the disposal and the method of steering a state by $n$ settings, then one can perform this method with use of $n$ from $n+1$ available settings.
\item See \cite{calvacanti-2017}.
\item This follows from Definitions \ref{def-entanglement} and \ref{def:classical-quantum}.
\item This follows from Definitions \ref{def-entanglement} and \ref{def:quantum-classical}.
\item This follows from Definitions \ref{def:classical-classical} and \ref{def:quantum-classical}.
\item This follows from Definitions \ref{def:classical-classical} and \ref{def:classical-quantum}.
\item This is a consequence of Theorem \ref{thm:super-discord-vs-discord} \cite{super-discord-1}.
\end{enumerate}
\end{proof}

The implications in Theorem \ref{thm-implications-1} do not hold the other way around.

\begin{theorem} \label{thm-implications-2}
For any bipartite state $\rho$ the following implications hold: 
\begin{enumerate}
\item $\rho$ is a product state, i.e. $\rho$ has zero super quantum discord $D_w (A|B) = D_w (B|A) =0$ $\Rightarrow$ $\rho$ is classical-classical, i.e. $\rho$ has both quantum discords zero $D(A|B) = B(B|A) =0$,
\item $\rho$ is classical-classical $\Rightarrow$ $\rho$ is both classical-quantum $D(B|A) =0$ and quantum-classical $D(A|B) =0$,
\item $\rho$ is either classical-quantum or quantum-classical $\Rightarrow$ $\rho$ is separable,
\item $\rho$ is separable $\Rightarrow$ $\rho$ is unsteerable (for any number of settings),
\item $\rho$ is unsteerable with $n+1$ settings $\Rightarrow$ $\rho$ is unsteerable with $n$ settings,
\item $\rho$ is unsteerable (with any number of settings) $\Rightarrow$ $\rho$ is local.
\end{enumerate}
\end{theorem}

\begin{proof}
This theorem follows easily from and Theorem \ref{thm-implications-1} and the respective definitions.
\end{proof}

\begin{figure}[H]
\centering
\includegraphics[scale=0.8]{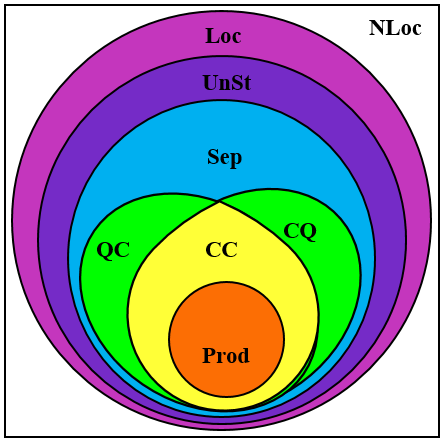}
\caption{Relations between properties of bipartite quantum states: Prod -- product states, CC -- classical-classical states, CQ -- classical-quantum states, QC -- quantum-classical states, Sep -- separable states, UnSt -- unsteerable states, Loc -- local states, NLoc -- nonlocal states. We have mentioned that the states with zero discord are of measure zero in the set of all density matrices, so the figure is out of scale.}
\label{fig:relations-general2}
\end{figure}

\begin{figure}[H]
\centering
\includegraphics[scale=0.8]{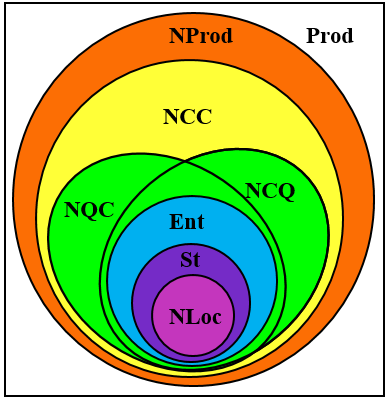}
\caption{Relations between properties of bipartite quantum states: NLoc -- nonlocality, St -- steerability, Ent -- entanglement, NCC -- states that are not classical-classical, NCQ -- states that are not classical-quantum, NQC -- states that are not quantum-classical, NonProd -- states that are not product, Prod -- product states. As before, the figure is out of scale.}
\label{fig:relations-general1}
\end{figure}

The property of negative/non-negative conditional entropy does not belong to the above hierarchy, because the following theorem holds:
\begin{theorem} [\citealp{friis}] \label{nonlocality-entropy}
In general nonlocality does not imply negative conditional entropy and negative conditional entropy does not imply nonlocality. 
\end{theorem}

We can see that this is true by looking at Table \ref{table-non-absolute} and Fig. \ref{fig:gisin}. For the Werner states there is an implication from negative conditional entropy to nonlocality.  
However, this is only a special case, as there are nonlocal Gisin states with positive conditional entropy and Gisin states with negative conditional entropy which are local.

However, negative conditional entropy requires entanglement:
\begin{theorem} [\citealp{cerf-adami}]
All states with negative conditional entropy are entangled.
\end{theorem}

It is known that for pure states some of the analysed properties become equivalent:
\begin{theorem} 
If $\rho$ is a pure bipartite state, the following equivalences hold: $\rho$ is nonlocal $\Leftrightarrow$ $\rho$ is steerable $\Leftrightarrow$ $\rho$ is entangled $\Leftrightarrow$ $\rho$ has negative conditional entropy. Equivalently: $\rho$ is local $\Leftrightarrow$ $\rho$ is not steerable $\Leftrightarrow$ $\rho$ is separable $\Leftrightarrow$ $\rho$ has non-negative conditional entropy.
\end{theorem}

The only property, which was not mentioned yet in the above theorems, is contextuality. It is different from the other properties because its definition does not rely on the division of physical system into subsystems. In fact, there are contextual states even in 3 dimensions, whereas the other properties are defined for systems of dimension at least 4 (the dimension cannot be a prime number). Therefore, for sure contextuality does not collapse to any other property described here. What is more, the set of contextual states does not contain and is not contained in any set of states possessing one of the other properties. However, one may ask what is the relation between contextuality and other properties in spaces where all of them are well-defined and, as far as I know, this relation has not been investigated.

\chapter{Absolute properties of quantum states}

Each of the properties defined in the previous chapter can be possessed by a given quantum state absolutely or non-absolutely. A property is possessed by a given state \textbf{absolutely} iff it is preserved under arbitrary unitary operation. Otherwise it is possessed \textbf{non-absolutely}. In each pair of properties usually exactly one of them can be possessed absolutely, e.g. there exist states absolutely separable, but there are no states absolutely entangled. The only exception is contextuality and noncontextuality. In this chapter we will put forward necessary and sufficient conditions for a given state possessing a given property absolutely (as far as such conditions are known). We will start from stating that usually only \textbf{global} unitary transformations matter in this context. The case of contextuality and noncontextuality is omitted here and will be considered in Section \ref{absolute-noncontextuality}.

\begin{fact} \label{local-unitary-does-not-change}
For any bipartite state $\rho$, none of the following properties of $\rho$: locality/nonlocality, steerability/unsteerability, entanglement/separability, the value of quantum discord and the value of quantum super discord can be changed by performing \textbf{local} unitary transformations.
\end{fact}



\section{Absolute separability} \label{absolute-separability}

The definitions of separability \ref{def-separability} and entanglement \ref{def-entanglement} assume a particular choice of factorisation of the Hilbert space. As a consequence, a state which is entangled with respect to a given choice of both subsystems can be separable for another choice of both subsystems. Therefore, one can formulate the following definition of absolutely separable states:

\begin{definition} [\citealp{kus-zyczkowski-2001}]
A bipartite state $\rho$ is called absolutely separable iff for any unitary operator $U$, the state $\rho' = U \rho U^\dagger$ is separable.
\end{definition}

The following three theorems describe the conditions under which a quantum state is absolutely separable:

\begin{theorem} [\citealp{thirring2011}] \label{thm:absolute-separability-pure}
Any separable pure state can be transformed by unitary operation into entangled state and the other way around. Therefore no pure separable states are absolutely separable.
\end{theorem}



\begin{theorem} [\citealp{verstraete-2001}] \label{thm:absolute-separability-mixed}
If $\rho$ is a mixed two-qubit state with an ordered spectrum $d_1 \geq d_2 \geq d_3 \geq d_4$, then $\rho$ is absolutely separable iff 
\begin{equation}
d_1 - d_3 - 2 \sqrt{d_2 d_4} \leq 0.
\end{equation}
\end{theorem}

The next theorem provides a generalisation of this result for a higher dimension of the second subsystem, where the first subsystem remains 2-dimensional:
\begin{theorem} [\citealp{hildebrand-2007} for $2 \times 3$ case, \citealp{johnston-2013} for the remaining cases] \label{thm:absolute-separability-higherdim} 
If $\rho$ is a bipartite state of dimension $2 \times n$ (for arbitrary $n$) with an ordered spectrum $d_1 \geq d_2 \geq \ldots \geq d_{2n}$, then $\rho$ is absolutely separable iff
\begin{equation}
d_1 - d_{2n-1} - 2 \sqrt{d_{2n-2} d_{2n}} \leq 0
\end{equation}
\end{theorem}

As we have seen, there is an asymmetry between separability and entanglement. Any state is separable in some basis (so there are no absolutely entangled states), but there are some states which are absolutely separable.

\section{Absolute PPT (Positive Partial Transpose)} \label{absolute-ppt}

For the PPT property one can formulate analogous notion of absolute PPT property as below:

\begin{definition}
A bipartite state $\rho $ is called absolutely PPT (Positive Partial Transpose) iff for any unitary operator $U$, the state $\rho' = U \rho U^\dagger$ satisfies PPT criterion \ref{criterion-ppt}.
\end{definition}

To give the necessary and sufficient conditions for a state's being absolutely PPT, let us define some auxiliary objects. We assume that the dimension of Hilbert space is $m \times n$. Define $p = \text{min} (m, n)$, $p_+ = p(p+1)/2$, $p_- = p(p-1)/2$, $S_+ = \{ (k,l) | 1 \leq k \leq l \leq p \}$, $S_- = \{ (k,l) | 1 \leq k < l \leq p \}$. Then, we define orderings of $S_+$ and $S_-$, namely $\sigma_+ : S_+ \mapsto \{ 1, \ldots, p_+ \}$ and $\sigma_- : S_- \mapsto \{ 1, \ldots, p_ \}$. This enables us to formulate the following theorem:

\begin{theorem} [\citealp{hildebrand-2007}] \label{thm:absolute-ppt} 
A bipartite state $\rho \in \tilde{\mathcal{H}}_m \otimes \tilde{\mathcal{H}}_n$ is absolutely PPT iff for any $\sigma_+, \sigma_-$, the following inequality holds:
\begin{equation}
\Lambda (\lambda ; \sigma_+ , \sigma_-) + \Lambda (\lambda ; \sigma_+, \sigma_-)^\top \geq 0,
\end{equation}
where $\Lambda$ is a matrix with elements
\begin{equation}
\Lambda_{kl} (\lambda ; \sigma_+ , \sigma_-) = \left\{
\begin{array}{c}
\lambda_{nm + 1 - \sigma_+ (k,l)} , \quad k \leq l,\\
- \lambda_{\sigma_- (l,k)}, \quad \quad \quad k > l, \\
\end{array}
\right.
\end{equation}
and $\lambda_1, \ldots, \lambda_{mn}$ are the eigenvalues of $\rho$ in decreasing order, assembled into a vector $\lambda$.
\end{theorem}

For $2 \times n$ case absolute separability and absolute PPT are equivalent \cite{johnston-2013} and in fact the Theorem \ref{thm:absolute-separability-higherdim} concerning absolute separability for these cases were proven with use of Theorem \ref{thm:absolute-ppt}.

\section{Absolute locality} \label{absolute-locality}

As in the previous cases, one can formulate the following definition of absolute version of the locality property:

\begin{definition}
A bipartite state $\rho$ is called absolutely local iff for any unitary operator $U$, the state $\rho' = U \rho U^\dagger$ is local.
\end{definition}

The necessary and sufficient conditions for a state being absolutely local have been found:

\begin{theorem}[\citealp{hindusi-locality}] \label{thm-absolute-locality}
If $\rho$ is a two-qubit state with an ordered spectrum $d_1 \geq d_2 \geq d_3 \geq d_4$, then $\rho$ is absolutely local iff
\begin{equation}
(2 d_1 + 2 d_2 - 1)^2 + (2 d_1 + 2 d_3 - 1)^2 \leq 1
\end{equation}
\end{theorem}

In the proof of this theorem the following lemma is used (it will be also needed in our proof in Section \ref{absolute-discord}):

\begin{lemma}[Cartan decomposition of $SU(4)$, \citealp{kraus-2001}]
Every matrix belonging to $SU(4)$ can be decomposed into two local unitary matrices $U_A \otimes U_B, V_A \otimes V_B$ and a global unitary matrix $U_g$ in the so-called Cartan form:
\begin{equation} \label{eq:U}
U = \left( U_A \otimes U_B \right) U_g \left( V_A \otimes V_B \right),
\end{equation}
where $U_A, U_B, V_A, V_B \in SU(2)$ and $U_g$ is given by
\begin{equation} \label{eq:Ug}
U_g = e^{-i (\lambda_1 \sigma_1 \otimes \sigma_1 + \lambda_2 \sigma_2  \otimes \sigma_2 + \lambda_3 \sigma_3 \otimes \sigma_3)}
\end{equation}
with $\lambda_1 , \lambda_2 , \lambda_3 \in [0, 2 \pi]$.
\end{lemma}

\section{Absolutely unsteerable states} \label{absolute-unsteerable}

Similarly to absolute locality, one can define absolute unsteerability:

\begin{definition}
A bipartite state $\rho$ is called absolutely unsteerable iff for any unitary operator $U$, the state $\rho' = U \rho U^\dagger$ is unsteerable.
\end{definition}

The necessary and sufficient conditions for a state being absolutely unsteerable are not known. However, there is in the literature the following partial result, concerning only steerability with three settings:

\begin{theorem}[\citealp{hindusi-steering}] \label{thm-absolute-steering}
If $\rho$ is a two-qubit state with spectrum $d_1 d_2, d_3, d_4$, then $\rho$ is absolutely unsteerable with three settings iff its eigenvalues satisfy
\begin{equation}
3 (d_1^2 + d_2^2 + d_3^2 +d_4^2) - 2(d_1 d_2 + d_1 d_3 + d_1 d_4 + d_2 d_3 + d_2 d_4 + d_3 d_4) \leq 1.
\end{equation}
\end{theorem}


\section{Absolute non-negativity of conditional entropy} \label{absolute-entropy}

Similarly to the previous properties, non-negative conditional entropy can also be possessed absolutely:

\begin{definition}[\citealp{hindusi-entropy}]
A bipartite state $\rho$ is said to have the conditional entropy non-negative absolutely iff for any unitary operator $U$, the state $\rho' = U \rho U^\dagger$ has non-negative conditional entropy.
\end{definition}

The necessary and sufficient conditions for absolute non-negativity of conditional entropy are known:

\begin{theorem}[\citealp{hindusi-entropy}] \label{thm-absolute-entropy}
A two-qubit state $\rho$ has absolutely non-negative conditional entropy iff $S (\rho) \geq 1$.
\end{theorem}

Recall that everywhere in this thesis '$\log$' means logarithm to base 2.



\section{Absolute zero quantum discord} \label{absolute-discord}

In contrast to the properties analysed before, quantum discord is quantitative --- we know its precise value, not only that it is possessed or not. However, one can 'discretise' this issue by dividing states into discordless (with quantum discord equal to zero) and states with non-zero quantum discord.
Therefore, we can ask when the value of quantum discord is absolutely zero, i.e. when a state does not have correlations of this type absolutely:

\begin{definition}
A bipartite state $\rho$ is said to have zero quantum discord absolutely iff for any unitary operator $U$, the state $\rho' = U \rho U^\dagger$ has zero quantum discord.
\end{definition}

In this section we will use Criterion \ref{criterion-zero-discord-2} \cite{huang2011} of zero quantum discord from Section \ref{zero-discord} to determine the necessary and sufficient conditions for having zero discord absolutely by a two-qubit states. For two-qubit system represented by a density matrix

\begin{equation}
\rho =  \left(
\begin{array}{cccc}
 r_{11} & r_{12} & r_{13} & r_{14} \\
r_{21} & r_{22} & r_{23} & r_{24} \\
r_{31} & r_{32} & r_{33} & r_{34} \\
 r_{41} & r_{42} & r_{43} & r_{44} \\
\end{array}
\right) 
\equiv \left(
\begin{array}{cc}
A & B \\
C & D \\
\end{array}
\right) 
\end{equation}
with blocks
\begin{equation}
A = \left(
\begin{array}{cc}
 r_{11} & r_{12} \\
r_{21} & r_{22} \\
\end{array}
\right), \quad
B = \left(
\begin{array}{cc}
r_{13} & r_{14} \\
r_{23} & r_{24} \\
\end{array}
\right), \quad
C = \left(
\begin{array}{cccc}
r_{31} & r_{32} \\
 r_{41} & r_{42} \\
\end{array}
\right), \quad
D = \left(
\begin{array}{cccc}
r_{33} & r_{34} \\
r_{43} & r_{44} \\
\end{array}
\right)
\end{equation}
this criterion means that the following equalities must be satisfied: 
\begin{equation} \label{eq:condition-zero-discord-1}
[A, A^\dagger] = [B, B^\dagger ] = [C, C^\dagger] = [D, D^\dagger] = 0
\end{equation}
and
\begin{equation} \label{eq:condition-zero-discord-2}
[A,B] = [A, C] = [A, D] = [B, C] = [B, D] = [C,D] = 0.
\end{equation}

The first simplification follows from the fact that every density matrix is hermitian and therefore can be diagonalized by some unitary matrix. If a given state has zero discord absolutely, in particular it has zero discord after diagonalization (because diagonalizing matrix belongs to the class of unitary matrices). Therefore, each equivalence class of states has a representative that is a diagonal matrix and to find the class of all states with absolute zero discord it suffices to restrict to the class of diagonal density matrices. 

Recall from section \ref{absolute-locality} that every $SU(4)$ matrix can be decomposed into local part and a special global matrix $U_g$ given by \eqref{eq:Ug}. One idea is to act with the global unitary matrix \eqref{eq:U} in its most general form and then solve the equations that follow from the conditions \eqref{eq:condition-zero-discord-1} and \eqref{eq:condition-zero-discord-2}. However, these equations are rather complicated, so from the practical point of view it is better to divide our task into two steps. In the first step we act on arbitrary diagonal matrix with $U_g$ only, obtaining necessary conditions for having zero discord absolutely. As we will see, the result will be one-parameter family of states. In the second step we will act on this family with $U$ in its general form (including local parts and the special global part). It turns out that it is possible to make further simplification by putting some parameters in local matrices to zero and such less general form is sufficient to eliminate all potential candidates for being absolutely zero discord state with the exception of one --- the maximally mixed state $\frac{1}{4}\mathds{1}_4$. It is easy to see that this state indeed has zero discord absolutely because after an arbitrary unitary transformation it remains unchanged.

Let us perform the first step. Consider an arbitrary diagonal density matrix
\begin{equation} \label{eq:diagonal-discord}
\rho_d =  \left(
\begin{array}{cccc}
d_1 & 0 & 0 & 0 \\
0 & d_2 & 0 & 0 \\
0 & 0 & d_3 & 0 \\
 0 & 0 & 0 & 1 - d_1 - d_2 - d_3 \\
\end{array}
\right),
\end{equation}
where $d_1, d_2, d_3 \in \mathds{R}$. The unitary matrix $U_g$ can be written in the Cartan form \eqref{eq:Ug}, which simplifies to
\begin{equation}
U_g = \left(
\scalemath{0.8}{\begin{array}{cccc}
 e^{-i \lambda _3} \cos \left(\lambda _1-\lambda _2\right) & 0 & 0 & -i e^{-i
   \lambda _3} \sin \left(\lambda _1-\lambda _2\right) \\
 0 & e^{i \lambda _3} \cos \left(\lambda _1+\lambda _2\right) & \sin
   \left(\lambda _1+\lambda _2\right) \left(\sin \lambda _3-i \cos
  \lambda _3 \right) & 0 \\
 0 & \sin \left(\lambda _1+\lambda _2\right) \left(\sin \lambda
   _3 -i \cos \lambda _3\right) & e^{i \lambda _3} \cos
   \left(\lambda _1+\lambda _2\right) & 0 \\
 -i e^{-i \lambda _3} \sin \left(\lambda _1-\lambda _2\right) & 0 & 0 & e^{-i
   \lambda _3} \cos \left(\lambda _1-\lambda _2\right) \\
\end{array}}
\right).
\end{equation}

Under the action of $U_g$ the state $\rho_d$ is transformed as follows:
\begin{equation}
\begin{split}
\rho'_d & = U_g^{\dagger} \rho_d U_g \\
& = \frac{1}{2}
\left(
\scalemath{0.8}{\begin{array}{cccc}
(2 d_1+d_2+d_3-1) C_- -d_2-d_3+1 & 0 & 0 & i (2 d_1+d_2+d_3-1) S_- \\
 0 & (d_2-d_3) C_+ +d_2+d_3 &   i (d_2-d_3) S_+ & 0 \\
 0 & -i (d_2-d_3) S_+ &   (d_2-d_3) C_+ +d_2+d_3 & 0   \\
 - i (2 d_1+d_2+d_3-1) S_- & 0   & 0 & -(2 d_1+d_2+d_3-1) C_- -d_2-d_3+1 \\
\end{array}}
\right), \end{split}
\end{equation}
where $S_+ = \sin (2 \lambda_1 + 2 \lambda_2)$, $S_{-} = \sin (2 \lambda_1 + 2 \lambda_2)$, $C_+ = \cos (2 \lambda_1 + 2 \lambda_2)$, $C_{-} = \cos (2 \lambda_1 - 2 \lambda_2)$.

To this transformed state $\rho'_d$ we apply conditions \eqref{eq:condition-zero-discord-1} and \eqref{eq:condition-zero-discord-2}. Three of them are always satisfied: $[A, A^\dagger] = [D, D^\dagger] = [A,D] =0$. The rest gives us equations for eigenvalues of $\rho_d$, which have the following solutions: $d_1 = d_2 = d_3 = \frac{1}{4}$ and $d_1 = \frac{1}{2} - d_2, d_3 = d_2$. This gives us the following necessary condition:
If a two-qubit state with eigenvalues $d_1, d_2, d_3, 1 - d_1 - d_2 - d_3$ has zero discord absolutely, then its eigenvalues satisfy the following relation:
\begin{equation} \label{eq:absolute-discord-2}
d_1 = \frac{1}{2} - d_2, d_3 = d_2
\end{equation}
or some of its permutations.

Now, let us perform the second step. Any unitary matrix belonging to $SU (2)$ can be parameterised in the following way:
\begin{equation} \label{eq:su2}
U_{\text{loc}} = \left(
\begin{array}{cc}
e^{i \alpha} \cos \phi & e^{i \beta} \sin \phi \\
- e^{- i \beta} \sin \phi & e^{- i \alpha} \cos \phi \\
\end{array}
\right)
\end{equation}

Each of the matrices $U_{\text{A}}, U_{\text{B}}, V_{\text{A}}, V_{\text{B}}$ has this form, so the whole matrix $U$ given by \eqref{eq:U} contains four independent matrices of the type \eqref{eq:su2}. We will add to each parameter $\alpha, \beta, \phi$ indices connected with matrices $U_{\text{A}}, U_{\text{B}}, V_{\text{A}}, V_{\text{B}}$, so e.g.
\begin{equation} \label{eq:UA}
U_{\text{A}} = \left(
\begin{array}{cc}
e^{i \alpha_{\text{UA}}} \cos \phi_{\text{UA}} & e^{i \beta_{\text{UA}}} \sin \phi_{\text{UA}} \\
- e^{- i \beta_{\text{UA}}} \sin \phi_{\text{UA}} & e^{- i \alpha_{\text{UA}}} \cos \phi_{\text{UA}} \\
\end{array}
\right)
\end{equation}
and similarly for $U_{\text{B}}, V_{\text{A}}$ and $V_{\text{B}}$.

It is sufficient to consider the case $\alpha_{\text{UA}} = \alpha_{\text{UB}} = \alpha_{\text{VA}} = \alpha_{\text{VB}} = 
\beta_{\text{UA}} = \beta_{\text{UB}} = \beta_{\text{VA}} = \beta_{\text{VB}} =
\phi_{\text{UB}} = \phi_{\text{VA}} = 0$ (only $\phi_{\text{UA}}$ and $\phi_{\text{VB}}$ are non-zero). We apply transformation of this type to our state \eqref{eq:diagonal-discord} satisfying \eqref{eq:absolute-discord-2}. From the conditions \eqref{eq:condition-zero-discord-1} and \eqref{eq:condition-zero-discord-2} we again obtain the set of equations constraining $d_2$, the only solution of which is the state $\frac{1}{4}\mathds{1}_4$. Therefore the following theorem holds:

\begin{theorem} \label{thm-absolute-zero-discord}
The only two-qubit state that has zero discord absolutely is the maximally mixed state $\frac{1}{4}\mathds{1}_4$.
\end{theorem}

Taking into account Theorem \ref{thm:classicality-vs-discord}, one can define absolute versions of being classical-classical, classical-quantum and quantum-classical:
\begin{definition} 
A bipartite state $\rho$ is called absolutely classical-classical/classical-quantum/quantum-classical iff for any unitary operator $U$, the state $\rho' = U \rho U^\dagger$ is classical-classical/classical-quantum/quantum-classical, respectively.
\end{definition}

The criteria for belonging to these classes of states are the same as for having vanishing discord: for absolutely quantum-classical $D (A|B)$ must vanish absolutely, for absolutely classical-quantum $D(B|A)$ must vanish absolutely and for absolutely classical-classical both of these conditions are needed. Of course, as a corollary to the previous theorem, the following holds for two-qubit states:
\begin{theorem} \label{thm-absolute-cc-cq-qc}
$ $

\noindent
The only two-qubit state that is absolutely classical-classical is the maximally mixed state $\frac{1}{4}\mathds{1}_4$. 

\noindent
The only two-qubit state that is absolutely classical-quantum is the maximally mixed state $\frac{1}{4}\mathds{1}_4$. 

\noindent
The only two-qubit state that is absolutely quantum-classical is the maximally mixed state $\frac{1}{4}\mathds{1}_4$. 
\end{theorem}

\section{Absolute zero quantum super discord}\label{absolute-super-discord}

Similarly to the case of quantum discord, we can define absolute version of quantum super discord:
\begin{definition}
A bipartite state $\rho$ is said to have zero quantum super discord absolutely iff for any unitary operator $U$, the state $\rho' = U \rho U^\dagger$ has zero super quantum discord.
\end{definition}
The Theorem \ref{thm:super-discord-product} \cite{super-discord-2} implies that having zero super discord absolutely is equivalent to being absolutely product, where the last property is defined as follows:
\begin{definition}
A bipartite state $\rho$ is said to be an absolutely product state iff for any unitary operator $U$, the state $\rho' = U \rho U^\dagger$ is a product state.
\end{definition}

From the Theorems \ref{thm:super-discord-vs-discord} \cite{super-discord-1} and \ref{thm-absolute-zero-discord} we can conclude that
\begin{theorem} \label{thm-absolute-zero-super-discord}
The only two-qubit state that has zero super discord absolutely is the maximally mixed state $\frac{1}{4}\mathds{1}_4$.
\end{theorem}

What is more, from the relation between between zero super discord and being a product state it follows, as a corollary, that
\begin{theorem} \label{thm-absolute-product}
The only two-qubit state that is absolutely product is the maximally mixed state $\frac{1}{4}\mathds{1}_4$. 
\end{theorem}


\section{Absolute contextuality and noncontextuality}\label{absolute-noncontextuality}

Contextuality and non-contextuality is the only pair of properties such that both elements of the pair have non-trivial absolute counterparts. 

\begin{definition}
A bipartite state $\rho$ is called absolutely contextual/absolutely noncontextual iff for any unitary operator $U$, the state $\rho' = U \rho U^\dagger$ is contextual/noncontextual, respectively.
\end{definition}

Contextuality and noncontextuality are different from other properties considered in this thesis because they do not refer to the division of system into subsystems. Therefore, we can expect that the distinction between local and global unitary operations does not matter for preserving this properties. In fact, one can prove even stronger result:

\begin{theorem} \label{thm:absolute-contextuality}
Contextuality and noncontextuality are always absolute, i.e. if a given state $\rho$ (of arbitrary dimensionality) is contextual, then it is also absolutely contextual and if it is noncontextual, then it is also absolutely noncontextual.
\end{theorem}

\begin{proof}
Suppose that for a given $\rho$ there exist projectors $P_i$, $i = 1, \ldots, 5$ such
that $P_i P_{i+1} = 0$ and Tr$(\rho P) > 2$, where $\mathfrak{P} = \sum_{i=1}^{5}
P_i$. Consider rotation of $\rho$ by an arbitrary unitary matrix $U$:
$\rho' = U \rho U^\dagger$. It can be shown that for this new state
also exist projectors which are witnesses for violating of KCBS
inequality: it suffices to take $P'_i = U P_i U^\dagger$, $\mathfrak{P}' =
\sum_{i=1}^{5} P'_i = U \mathfrak{P} U^\dagger$. They satisfy the condition of
orthogonality $P'_i P'_{i+1} = U P_i U^\dagger U P_{i+1} U^\dagger = U P_i P_{i+1}
U^\dagger = 0$ and Tr$(\rho' \mathfrak{P}') =$ Tr$(U \rho U^\dagger U \mathfrak{P} U^\dagger )$ = Tr$(U \rho \mathfrak{P}
U^\dagger)$ = Tr$(U^\dagger U \rho \mathfrak{P})$ = Tr$(\rho \mathfrak{P}) >2$.
\end{proof}

One may wonder why do we not have a similar argument for nonlocality, as both contextuality and nonlocality consist of violation of certain inequality for some choice of appropriate operator ($\mathfrak{B}_{CHSH}$ or $\mathfrak{P}$, respectively). The difference lies in dissimilar criteria imposed on these operators. The CHSH operator must have a certain structure given by the equation \eqref{eq:CHSH-operator}. This structure can be spoilt by global unitary operation, whereas, as we have seen, the conditions defining $\mathfrak{P}$ are still satisfied after an arbitrary unitary operation.

\chapter{Relations between different absolute properties} \label{sec-relations-absolute}

\section{Relations for special classes of states} \label{sec-relations-absolute-special}

In Table \ref{table-absolute} there are shown ranges of parameters for which the Werner states \eqref{eq:werner} and the Gisin states \eqref{eq:gisin} are absolutely product states, have zero discord absolutely, are absolutely separable, absolutely unsteerable, absolutely local and have non-negative conditional entropy absolutely. These results have been obtained with the use of the following Theorems: \ref{thm:absolute-separability-mixed} \cite{verstraete-2001}, \ref{thm-absolute-locality} \cite{hindusi-locality}, \ref{thm-absolute-steering} \cite{hindusi-steering}, \ref{thm-absolute-entropy} \cite{hindusi-entropy}, \ref{thm-absolute-zero-discord} and \ref{thm-absolute-zero-super-discord}. 
Note that for the Gisin states the parameter $\theta$ does not matter in this context; this is because possessing absolute properties depends only on eigenvalues of the state and for the Gisin states eigenvalues depend only on parameter $\lambda$ (see Section \ref{sec:gisin}). Some of these results were already presented in the literature: \citealp{hindusi-locality} (absolute separability and absolute locality for the Werner states, absolute locality for the Gisin states), \citealp{hindusi-entropy} (absolutely non-negative conditional entropy for the Werner states), \citealp{hindusi-steering} (absolute unsteerability for the Werner states and the Gisin states). 

When we compare these results with Table \ref{table-non-absolute} as well as Fig. \ref{fig:werner} and \ref{fig:gisin}, we can observe that for the Werner states there is no difference between possessing a given property and possessing a given property absolutely. In contrast, for the Gisin states the ranges of parameters are changed in all the cases\footnote{We do not know the ranges of the parameters for which the Gisin states are unsteerable. However, we know that the set of unsteerable states must contain the set of separable states (Theorem \ref{thm-implications-2}) and that the set of absolutely unsteerable states must be contained in the set of absolutely local states (Theorem \ref{thm-absolute-implications}). As a consequence, if the sets of Gisin unsteerable states and of Gisin absolutely unsteerable states were the same, the set of separable Gisin states would be contained in the set of absolutely local Gisin states. As there are Gisin states which are separable but not absolutely local, the sets of unsteerable states and absolutely unsteerable states are not the same.}. Therefore, the equivalence between the properties and the respective absolute properties holds only for very special families of states (such as the Werner states) and in general is not true. The comparison between the properties and the respective absolute properties for Gisin states is illustrated in Fig. \ref{fig:gisin-absolute}. 

The results for the Werner states and the Gisin states allow one to distinguish between almost all of the absolute properties analysed here: only being an absolutely product state and having zero quantum discord absolutely are impossible to distinguish (and as we have seen in Section \ref{absolute-super-discord}, in fact they are generally equivalent).

\begin{table}[H]
\centering
{\def\arraystretch{2}\tabcolsep=5pt
\begin{tabular}{llll}
\hline
\multicolumn{1}{|l|}{}                                        & \multicolumn{1}{l|}{\textbf{Werner states} \eqref{eq:werner}}                               & \multicolumn{2}{l|}{\textbf{Gisin states} \eqref{eq:gisin}}                                                                                                                \\ \hline
\multicolumn{1}{|l|}{parameters}      & \multicolumn{1}{l|}{$w$}        & \multicolumn{1}{l|}{$\theta$}            & \multicolumn{1}{l|}{$\lambda$}        \\ \hline
\multicolumn{1}{|l|}{range of the parameter}  & \multicolumn{1}{l|}{$w \in [-\frac{1}{3},1]$}        & \multicolumn{1}{l|}{$\theta \in [0,\frac{\pi}{2}]$} & \multicolumn{1}{l|}{$\lambda \in [0,1]$}     \\ \hline
\multicolumn{1}{|l|}{absolutely product}                      & \multicolumn{1}{l|}{$w=0$}                                           & \multicolumn{1}{l|}{never}               & \multicolumn{1}{l|}{never}                                                                            \\ \hline
\multicolumn{1}{|l|}{absolutely zero discord}                 & \multicolumn{1}{l|}{$w=0$}                                           & \multicolumn{1}{l|}{never}               & \multicolumn{1}{l|}{never}                                                                            \\ \hline
\multicolumn{1}{|l|}{absolutely separable}                    & \multicolumn{1}{l|}{$w \in [-\frac{1}{3}, \frac{1}{3}]$}                 & \multicolumn{1}{l|}{never}               & \multicolumn{1}{l|}{never}                                                                            \\ \hline
\multicolumn{1}{|l|}{absolutely unsteerable}                  & \multicolumn{1}{l|}{$w \in [-\frac{1}{3}, \frac{1}{\sqrt{3}}]$} & \multicolumn{1}{l|}{$\theta$ arbitrary}           & \multicolumn{1}{l|}{$\lambda \in [0, \frac{2}{3}]$}                                                               \\ \hline
\multicolumn{1}{|l|}{absolutely local}                        & \multicolumn{1}{l|}{$w \in [-\frac{1}{3}, \frac{1}{\sqrt{2}}]$}         & \multicolumn{1}{l|}{$\theta$ arbitrary}           & \multicolumn{1}{l|}{$\lambda \in [0, \frac{1}{\sqrt{2}}]$}                                                        \\ \hline
\multicolumn{1}{|l|}{\parbox[t]{4cm}{absolutely non-negative \\ conditional entropy \\ $ $}} & \multicolumn{1}{l|}{$w \in [-\frac{1}{3}, w_1]$, $w_1 \approx 0.7476$} & \multicolumn{1}{l|}{$\theta$ arbitrary}           & \multicolumn{1}{l|}{$\lambda \in [0, \lambda_1]$, $\lambda_1 \approx 0.7729$} \\ \hline
\end{tabular}%
}
\caption{Selected absolute properties of the Werner states and the Gisin states in function of their parameters.}
\label{table-absolute}
\end{table}

\begin{figure}[H]%
\centering
\subfloat{{\includegraphics[scale=0.6]{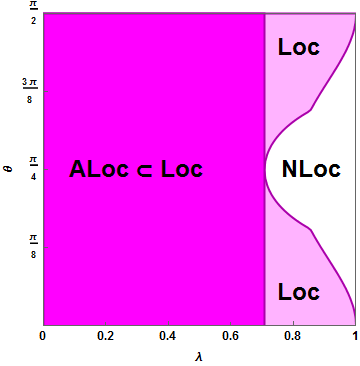} }}%
\qquad
\subfloat{{\includegraphics[scale=0.6]{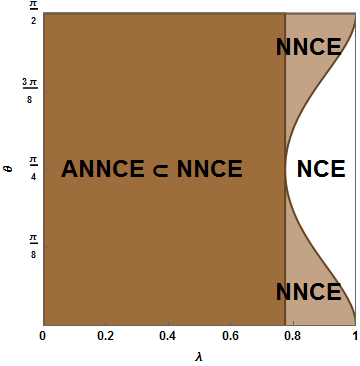} }}%
\caption{Properties of the Gisin states. Left figure: ALoc --  absolutely local states (dark purple), Loc -- local states (dark purple and light purple), NLoc -- nonlocal states (white). Right figure: ANNCE -- absolutely non-negative conditional entropy (dark brown), NNCE -- non-negative conditional entropy (light brown and dark brown), NCE -- negative conditional entropy (white).}
\label{fig:gisin-absolute}
\end{figure}

\section{General relations} \label{sec-relations-absolute-general}

The relations between absolute versions of the properties are very much similar to the relations between 'ordinary' versions of these properties, as summarised in the following theorem:

\begin{theorem} \label{relationAB}
Assume that for any quantum state possessing the property $A$ implies possessing the property $B$. It follows that if a given state $\rho$ has the property $A$ absolutely, then it also has the property $B$ absolutely.
\end{theorem}
\begin{proof}
Assume that $\rho$ has the property $A$ absolutely and that for any quantum state, possessing the property $A$ implies possessing the property $B$. Let us transform the state $\rho$ by some unitary operator $U$, obtaining $\rho' = U \rho U^\dagger$. Because $\rho$ has the property $A$ absolutely, $\rho'$ must have the property $A$. Therefore, from the implication, $\rho'$ must have the property $B$. As the unitary operator $U$ was arbitrary, it follows that $\rho$ has the property $B$ absolutely.
\end{proof}

From the Theorem \ref{relationAB} it follows that the Theorem \ref{thm-implications-2} is true for absolute versions of the properties as well:
\begin{theorem} \label{thm-absolute-implications} For any bipartite state $\rho$ the following implications hold: 
\begin{enumerate}
\item $\rho$ is a product state absolutely, i.e. $\rho$ has zero super quantum discord $D_w (A|B) = D_w (B|A) =0$ absolutely $\Rightarrow$ $\rho$ is absolutely classical-classical, i.e. $\rho$ has both quantum discords zero $D(A|B) = B(B|A) =0$ absolutely,
\item $\rho$ is absolutely classical-classical $\Rightarrow$ $\rho$ is both absolutely classical-quantum $D(B|A) =0$ and absolutely quantum-classical $D(A|B) =0$,
\item $\rho$ is either absolutely classical-quantum or absolutely quantum-classical $\Rightarrow$ $\rho$ is absolutely separable,
\item $\rho$ is absolutely separable $\Rightarrow$ $\rho$ is absolutely unsteerable (with any number of settings),
\item $\rho$ is absolutely unsteerable with $n+1$ settings $\Rightarrow$ $\rho$ is absolutely unsteerable with $n$ settings,
\item $\rho$ is absolutely unsteerable (with any number of settings) $\Rightarrow$ $\rho$ is absolutely local.
\end{enumerate}
\end{theorem}

The implication from absolute separability to absolute locality has been already noted in \cite{hindusi-locality-2} and \cite{hindusi-locality}. The relation between the above properties and absolute non-negativity of conditional entropy is less understood. Its relation with absolute separability is known:

\begin{theorem} [\citealp{hindusi-entropy}]
The class of absolutely separable two-qubit states forms a proper subset of the class of two-qubit states that have non-negative conditional entropy absolutely.
\end{theorem}

\begin{proof}
Being a subset follows from the Theorem \ref{relationAB}. Being a proper subset follows from the fact, that there exist two-qubit states which have non-negative conditional entropy absolutely but are not absolutely separable. Examples of such states are Gisin states for $\lambda \in [0, \lambda_1]$, $\lambda_1 \approx 0.7729$ and arbitrary $\theta$ (see Table \ref{table-absolute}).
\end{proof}

The relation between absolutely local two-qubit states and two-qubit states that have non-negative conditional entropy absolutely is in general not known. For the Werner states absolute separability implies absolute non-negative conditional entropy, as the Werner states are absolutely separable for $w \leq \frac{1}{3}$ and have non-negative conditional entropy absolutely for $w \leq w_e$, where $w_e $ is the solution of the equation $3(1-w_e) \log(1-w_e)+(1+3w_e) \log(1+3w_e) = 4$ and its numerical value is $w_e \approx 0.7476$ \cite{hindusi-entropy}. As everywhere else, by '$\log$' we mean here logarithm to base 2.

\begin{figure}[h]
\centering
\includegraphics[scale=0.8]{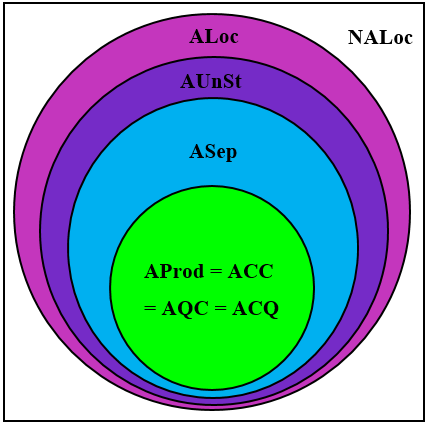}
\caption{Relations between different absolute properties: AProd -- absolutely product states, ACC -- absolutely classical-classical states (with both discords equal zero absolutely), ACQ -- absolutely classical-quantum states, AQC -- absolutely quantum-classical states, ASep -- absolutely separable states, AUnSt -- absolutely unsteerable states, ALoc -- absolutely local states. Note that the green circle denotes only one point ($\frac{1}{4}\mathds{1}_4$), so the figure is out of scale.}
\label{fig:absolute}
\end{figure}
\FloatBarrier

\chapter{Concluding remarks}

The aim of this thesis was to provide a review of properties of quantum states analysed in the literature, relations between them and their behaviour under unitary transformations. For each property X, its 'absolute' version (the property of being absolutely X) was introduced: a quantum state $\rho$ is said to be an absolutely X state iff for any unitary operator $U$, the state $\rho' = U \rho U^\dagger$ has the property X. 

The following new results have been proved:
\begin{itemize}
\item The only two-qubit state that has zero discord absolutely is the maximally mixed state $\mathds{1}_4$ --- see Theorem \ref{thm-absolute-zero-discord};
\item The only two-qubit state that has zero super discord absolutely is the maximally mixed state $\mathds{1}_4$ --- see Theorem \ref{thm-absolute-zero-super-discord});
\item Contextuality and noncontextuality are always absolute, i.e. if a given state $\rho$ (of arbitrary dimensionality) is contextual, then it is also absolutely contextual and if it is noncontextual, then it is also absolutely noncontextual --- see Theorem \ref{thm:absolute-contextuality};
\item If for any quantum state possessing the property $A$ implies possessing the property $B$, then also possessing the property $A$ absolutely implies possessing the property $B$ absolutely --- see Theorem \ref{relationAB}.
\end{itemize}

With regard to specific classes of states, for the Gisin states the range of parameters for which they are product, zero discord, absolutely product and absolutely zero discord have been determined --- see Table \ref{table-non-absolute} in Section \ref{sec-relations-special} and Table \ref{table-absolute} in Section \ref{sec-relations-absolute-special}. 

The presented results do not exhaust the topic. The following issues remain for further study:
\begin{itemize}
\item The theorems concerning absolute zero discord and absolute product states have been proven only for two-qubit states. The conjecture that they also hold for higher dimensions seems to be natural. Given the parameterisation of unitary matrix for the dimension $m \times n$, one can extend the method used in the proof of the Theorem \ref{thm-absolute-zero-discord} to check the conjecture for the $m \times n$ case. However, the number of equations to solve will be large and, what is worse, this method could not be used to confirm the conjecture in its full generality. Therefore, some other methods are required.
\item There is a hierarchy between properties of quantum states described in Theorems \ref{thm-implications-1}, \ref{thm-implications-2} and \ref{thm-absolute-implications}. The natural question to ask is whether one can find some other property that can be included into this hierarchy.
\item The property of (non-)negative conditional entropy does not belong to the hierarchy because of its relation with nonlocality (see Theorem \ref{nonlocality-entropy}, \citealp{friis}). Its relation with some other properties is not known, both for the absolute and for the 'ordinary' case, so this is another issue for future investigations.
\item The notion of nonlocality analysed here is not the most general one: we have discussed only CHSH nonlocality, which is of the (2, 2, 2) type (two subsystems, two measurements and two outcomes). The reason for this restriction is the fact that only for this case a general and easy-to-use criterion is known (Criterion \ref{criterion-chsh}, \citealp{horodeccy-1995}). It would be interesting to analyse the problem in full generality and see the relations between other types of locality/nonlocality and the rest of properties analysed in this thesis, both for the absolute and for the 'ordinary' case.
\end{itemize}

\chapter{Acknowledgements}

I would like to thank my supervisor, prof. Karol Życzkowski, for his help in preparing this thesis. I want also thank Dardo Goyeneche and Konrad Szymański for useful discussions.

 \end{document}